\documentclass[twocolumn]{IEEEtran}

\newcommand{\bzero}{\mbox{\boldmath{$0$}}}

\newcommand{\bA}{\mbox{\boldmath{$A$}}}
\newcommand{\ba}{\mbox{\boldmath{$a$}}}

\newcommand{\bc}{\mbox{\boldmath{$c$}}}

\newcommand{\ee}{\end{equation}}

\newcommand{\bI}{\mbox{\boldmath{$I$}}}

\newcommand{\bJ}{\mbox{\boldmath{$J$}}}

\newcommand{\bM}{\mbox{\boldmath{$M$}}}

\newcommand{\bn}{\mbox{\boldmath{$n$}}}
\newcommand{\bP}{\mbox{\boldmath{$P$}}}
\newcommand{\bp}{\mbox{\boldmath{$p$}}}
\newcommand{\bQ}{\mbox{\boldmath{$Q$}}}
\newcommand{\bq}{\mbox{\boldmath{$q$}}}
\newcommand{\bR}{\mbox{\boldmath{$R$}}}

\newcommand{\bs}{\mbox{\boldmath{$s$}}}

\newcommand{\bv}{\mbox{\boldmath{$v$}}}

\newcommand{\bw}{\mbox{\boldmath{$w$}}}
\newcommand{\bX}{\mbox{\boldmath{$X$}}}
\newcommand{\bx}{\mbox{\boldmath{$x$}}}

\newcommand{\by}{\mbox{\boldmath{$y$}}}

\newcommand{\bz}{\mbox{\boldmath{$z$}}}

\newcommand{\tr}{\mbox{\rm tr}\, }

\newcommand{\diag}{\mbox{\boldmath\bf diag}\, }
\newcommand{\Diag}{\mbox{\boldmath\bf Diag}\, }

\def\E{{\mathbb E}}

\usepackage[overload]{empheq}
\usepackage{fancyhdr}
\usepackage[colorlinks]{hyperref}

\usepackage{amsthm}
\usepackage[T1]{fontenc}
\usepackage[latin9]{inputenc}
\usepackage{bm}
\usepackage{amsmath}
\usepackage{amssymb}
\usepackage{amsfonts}
\usepackage{subfigure}
\usepackage{mathrsfs}
\usepackage{pifont}
\usepackage{amssymb}
\usepackage{verbatim}
\usepackage{upgreek}
\usepackage{color}
\usepackage[final]{graphicx}
\usepackage{epsfig}
\usepackage{booktabs}
\usepackage{bm}
\usepackage{setspace}
\usepackage{url}
\usepackage{algorithmic}
\usepackage{algorithm}
\usepackage{cite}
\usepackage{pifont}
\usepackage{bm}
\usepackage{setspace}
\usepackage{url}

\newtheorem{corollary}{Corollary}[section]

\newtheorem{lemma}{Lemma}[]
\newtheorem{proposition}{Proposition}[section] 
\newtheorem{remark}{Remark}[section]

\title{ Multi-Spectrally Constrained Transceiver Design against Signal-Dependent Interference }

\author{Jing~Yang, \IEEEmembership{Student Member, IEEE}, Augusto~Aubry, \IEEEmembership{Senior Member, IEEE},  Antonio~De~Maio, \IEEEmembership{Fellow, IEEE}, {Xianxiang~Yu}, and 	{Guolong~Cui}, \IEEEmembership{Senior Member, IEEE}
	\thanks{J.~Yang, X.~Yu and G.~Cui are with the School of Information and Communication Engineering,
		University of Electronic Science and Technology of China, Chengdu 611731,	China. E-mail: yangjinguestc@163.com, xianxiangy@gmail.com,  cuiguolong@uestc.edu.cn.
	}
	\thanks{A.~Aubry and A.~De~Maio are with the Department of Electrical and Information Technology Engineering, University of Naples Federico II, $80125$ Napoli, Italy.
		Email: augusto.aubry@unina.it, ademaio@unina.it.}
	\thanks{The work of A. Aubry and A. De Maio was supported by University of Naples Federico II-Radar wAveform Diversity for spectrAlly cRowded
		envirOnments based on oPtimization (RADAROPT) project, and by PON
		ARS01 00615 - OPL-APPS - IIoT OPEN Platform e Applicazioni per il
		manufactoring.
		The work of J. Yang, X. Yu and G. Cui was supported  by the National Natural Science Foundation of China under Grants 61771109, U19B2017, 61871080 and 61701088, by ChangJiang Scholar Program, by the 111 project No.B17008, by China Postdoctoral Science Foundation under Grant 2020M680147.}}

\begin{document}
\maketitle

\begin{abstract}
This paper focuses on the joint synthesis of constant envelope transmit signal and receive filter aimed at optimizing radar performance in signal-dependent interference and spectrally contested-congested environments. To ensure the desired Quality of Service (QoS) at each communication system, a precise control of the interference energy injected by the radar in each licensed/shared bandwidth is imposed.  Besides, along with an upper bound to the maximum transmitted energy, constant envelope (with either arbitrary or discrete phases) and similarity constraints are forced to ensure compatibility with amplifiers operating in saturation regime and bestow relevant waveform features, respectively.
To handle the resulting NP-hard design
problems, new iterative procedures (with ensured convergence properties) are devised to account for continuous and discrete phase constraints, capitalizing on the Coordinate Descent (CD) framework.  Two heuristic procedures are also proposed to perform valuable initializations.
   Numerical results are provided to assess the effectiveness of the conceived algorithms in comparison with the existing methods.
\end{abstract}


% \IEEEpeerreviewmaketitle

%%%%%%%%%%%%%%%%%%%%%%%%%%%%%%%%%%%%%%%%%%%%%%%%%%%%%%%%%%%%%%%%%%%%%%%%%%%%%%%%%%%%%%%%%%%%%%%%%%%%%
%%%%%%%%%%%%%%%%%%%%%%%%%%%%%%%%% Sections %%%%%%%%%%%%%%%%%%%%%%%%%%%%%%%%%%%%%%%%%%%%%%%%%%%%%%%%%%

\section{Introduction}
% no \IEEEPARstart

Spectral coexistence among radar and telecommunication systems has drawn flourishing attention due to the conflict between limited Radio Frequency (RF) spectrum resource and the increasing demand of spectrum access \cite{farina2017impact,maio2020radar,riff,Govoni,Deng}. Waveform diversity and cognitive radar are key candidates to alleviate this problem via on the fly adaptation of the transmitted waveforms to actual spectrally contested and congested environments\cite{speex,spm,wicks,cog,gini2008knowledge,blunt2016overview,Lops1,Lops2,Lops3,Lops4}.
In this respect, a plethora of papers in the open literature have dealt with the problem of designing cognitive radar signals with a suitable frequency allocation so as to induce acceptable interference levels on the frequency-overlaid systems, while improving radar performance  in terms of low range-Doppler sidelobes, detection, and tracking abilities\cite{Nunn,7961257,alhujaili2020spectrally,tang2018efficient,bicua2018radar,Ger1,Gerlach2,Pillai1,pillai2,Lindenfeld}.
%, as well as radar performance enhancement in terms of range-Doppler sidelobes, detection capability and so on.

%In \cite{5494625}, chirp-like constant modulus signals with a single null is synthesized, while extended to multiple notches in \cite{5960706}.
In \cite{he2010waveform}, a technique to synthesize constant envelope signals sharing a low Integrated Sidelobe Level (ISL) and a sparse spectral allocation is developed, considering as objective function a weighted sum between an ISL-oriented contribution and a term accounting for the waveform energy over the licensed bandwidths. The mentioned approach is generalized to provide a control on the waveform ambiguity function features in \cite{8454704}.
Some interesting algorithms to devise radar  waveforms under spectral compatibility requirements, are also proposed in  \cite{fan2020minimum,8967148,fan_spectrally_2021}, 
 where different performance criteria, like amplitude dynamic range, Peak Sidelobe Level (PSL), spectral shape features, and distance from a reference code, are considered at the signal design process.	
\cite{6784117} introduces the spectral shaping (SHAPE) method to synthesize constant-modulus waveforms aimed at fitting an arbitrary desired spectrum magnitude, whereas in \cite{8358735} and \cite{8770133} the Spectral Level Ratio (SLR)  is adopted as performance metric  assuming, respectively, constant envelope and PAR constraints.
In \cite{6850145, newsi, farina2017impact,tang2018alternating}, the maximization of Signal to Interference plus Noise Ratio (SINR) is accomplished in the presence of  signal-independent disturbance while controlling  the total interference energy on the common band and some desirable features of the transmitted waveform. 
This framework is extended to incorporate multiple spectral compatibility  constraints in\cite{7414411}. Along this line, to comply with the current amplifier technology, extensions to
address optimized synthesis of constant envelop waveforms with the continuous phase\cite{atsp} and the finite alphabet\cite{9082109}, are developed. However, the studies in \cite{6850145, newsi, farina2017impact,tang2018alternating,7414411,atsp,9082109} do not account for signal-dependent interference at the design stage, namely they implicitly assume that the radar is pointing toward the sky (with very low antenna sidelobes) or the target range is far enough that ground clutter is substantially absent. 
Some attempts to design radar transceivers capable of lifting up detection performance in a highly reverberating environment as well as ensuring spectral compatibility have been pursued in the open literature.
For instance, \cite{aa} and \cite{7838312} optimize the SINR (in the presence of signal-dependent disturbance) over the transmit signal and receive structure, while controlling the total amount of interference energy injected on the shared frequency bands.
Still forcing  a constraint on the global spectral interference, \cite{b2,yumo,8356676} propose waveform design procedures in the context of Multiple-Input Multiple-Output (MIMO) radar systems operating in highly reverberating environments. Besides, \cite{b2} also extends the developed framework considering multiple spectral constraints but just an energy constraint is forced on each transmitted signal.  Nevertheless, the design of constant envelope signals ensuring the appropriate Quality of Service (QoS) to each licensed system  still remains an open issue.

Aimed at filling this gap, in this paper,  a new radar transceiver design strategy (with phase-only probing signals) is proposed aimed at optimizing surveillance system performance (via SINR maximization in signal-dependent interference) while fully guaranteeing coexistence with the surrounding RF emitters. Specifically, unlike most of the previous works\footnote{Preliminary results of continuous phase codes are shown in \cite{myconf} without technical details.}, a local control on the interference energy radiated by the constant envelope signal on each reserved frequency bandwidth is performed, so as to enable joint radar and communication activities.  
Moreover,  to comply with the current amplifiers technology (operating in saturation regime) constant envelope waveforms are considered, with either arbitrary or discrete phases.
Besides, to fulfill basic radar requirements, in addition to an upper bound to the maximum radiated energy, a similarity constraint is enforced  to bestow relevant waveform hallmarks, i.e., a well-shaped ambiguity function. 
To handle the resulting NP-hard optimization problems, a suitable re-parameterization of the radar code vector is performed; hence, leveraging the Coordinate Descent (CD) method
\cite{CD}, iterative algorithms (monotonically improving the SINR) are proposed (for both continuous and discrete phase constraints), where either a specific entry of the transmitter parameter vector or the receive filter is optimized at a time while keeping fixed the other variables. Specifically, a global optimal solution of each, possibly non-convex, optimization problem (involved in the two developed procedures) is derived in closed form through the evaluation of elementary functions. Regardless of the phase cardinality, the computational complexity is linear with respect to the number of iterations, cubic with reference to the code length, and less than quadratic with respect to the number of spectral constraints. Finally, two heuristic approaches, accounting for spectral compatibility requirements via a penalty term in the
objective function, are proposed to initialize the procedures via ad-hoc starting solutions. 
To shed light on the capability of the devised algorithms to counter signal-dependent interference and ensure coexistence with the overlaid RF emitters, some case studies are provided at the analysis stage. Moreover, appropriate comparisons with some counterparts available in the open literatures are presented to prove the effectiveness of the new proposed strategies. 

The paper is organized as follows. In
Section \ref{sec_ProbForm},  the system model is introduced, followed by the definition and description of the key performance metrics as well as constraints involved into the formulation of radar transceiver design problems under investigation. In Section \ref{CD-based}, innovative solution methods are developed to handle the NP-hard optimization problems at hand. Section \ref{sec_performance} presents some numerical results to assess the performance. Finally, in Section \ref{sec_conclusion}, concluding remarks and some possible future research avenues are provided.

%\section*{Notations}{ We adopt the notation of using boldface for vectors $\ba$ (lower case), and matrices $\bA$ (upper case). 
%	The transpose and the conjugate transpose operators are denoted by $(\cdot)^T$ and $(\cdot)^H$, respectively.  ${\mathbb{C}}^N$ and ${\mathbb{R}}^N$ are, respectively, the sets of $N$-dimensional vectors  of complex and real numbers.  The letter $j$ represents the imaginary unit (i.e. $j=\sqrt{-1}$).  $\Re\{\cdot\}$, $\Im\{\cdot\}$ and  $|\cdot|$ mean, respectively, the real, imaginary part and modulus of a complex number.
%	Moreover, for any ${\bar \bs} \in \mathbb{C}^N$, $\|{\bar \bs}\|$ and $\|{\bar \bs}\|_{\infty}$ denote the Euclidean and $l$-infinity norm, respectively. $\diag({\bar \bs})$ indicates the diagonal matrix whose $i$-th diagonal element is the $i$-th entry of ${\bar \bs}$. $\E\left[\cdot\right]$ denotes statistical expectation and for any optimization problem $\cal{P}$, $v(\cal{P})$ is its optimal value. Finally, $\odot$ denotes the Hadamard element-wise product.
%}
\begin{table}[]
	\centering
	\caption{Notations} \label{table11}
	\hrulefill
	\begin{itemize}
		\item Bold letters, e.g., $\ba$ (lower
		case), and $\bA$ (upper case) denote vector and matrix, respectively.
		\item $(\cdot)^T$, $(\cdot)^*$, and $(\cdot)^\dag$ indicate
		the transpose, the conjugate, and the conjugate transpose operators, respectively.	
		\item  ${\mathbb{R}}^{N}$, ${\mathbb{C}}^{N}$, ${\mathbb{H}}^{N}$ are the sets of $N$-dimensional vectors of real and complex numbers, and of $N\times N$-dimensional Hermitian matrices, respectively.
		\item For any ${\bs} \in \mathbb{C}^N$, $\|{\bs}\|$ and $\|{\bs}\|_{\infty}$ represent the Euclidean and $l$-infinity norm, respectively.
		\item  $\bI_{N}$ and $\bzero_{N}$ represent the $N\times N$-dimensional identity matrix and 
		the matrix with zero entries.
		\item ${\bf 1}_{N}$ is the $N\times 1$-dimensional vector with all entries equal to 1.
		\item  $\bm e_{n}\in {\mathbb{R}}^{N}$ is a vector whose $n$-th entry is $1$ and other elements are $0$.
		\item   Letter $j$ represents the imaginary unit (i.e.,
		$j=\sqrt{-1}$).
		\item $\Re\{\cdot\}$, $\Im\{\cdot\}$ and  $|\cdot|$ mean the real, imaginary part, and modulus of a complex number, respectively.
		\item
		$\arg(x)\in[-\pi,\pi[$ represents the  argument of the complex number $x$.
		\item $\diag(\ba)$
		indicates the diagonal matrix formed by the entries of $\ba$.
		\item $\Diag(\bA)$ indicates the diagonal matrix whose $i$-th diagonal element is $\bA(i,i)$
		\item  ${{\bJ}_m}\in \mathbb{C}^N$ is the shift matrix with ${\bJ}_m(i,l)= 1$  if $i-l=m$, else ${\bJ}_m(i,l) = 0$, $i,l\in \{ 1,\ldots, N\}$.
		\item $\lambda_{\max }({\bA})$ is the largest eigenvalue of ${\bA}$.
		\item The statistical expectation is indicated as $\mathbb{E}\{\cdot\}$.
		\item $\lfloor{a}\rfloor$ and $\lceil{a}\rceil$ ($a\in\mathbb{R}$) provide the greatest integer not larger than $a$ and the lowest integer not smaller than $a$, respectively.
		\item $\lceil a\rfloor, a\in\mathbb{R}$ denotes the  integer closest to $a$. If the  decimal part of $a$ is 0.5, $\lceil a\rfloor=\left\lceil a \right\rceil \text{ if }a < 0$, else $\lceil a\rfloor=
				\left\lfloor a \right\rfloor $.
%		\item $\bx=\rm sort(\ba),\ba\in{\mathbb{R}}^{N}$ is the $N$-dimensional vector containing the entries of a sorted ascending order.
		\item $\odot$ is the Hadamard element-wise product.
		\item $v(\cal{P})$ is the optimal value of the optimization Problem $\cal{P}$.
		\item $\frac{\partial {f}(x)}{\partial x}$ denote the derivative of  ${f}(x)$ with respect to $x$.
%		\item $\vec(\bA)$ is the column vector obtained by stacking the columns of $\bA$.
%		\item ${\rm mat}_{(N,M)}(\ba)$ means the operation of reshaping the vector $\ba\in\mathbb{C}^{NM}$ into an $N\times M$ matrix, which satisfies ${\rm mat}_{(N,M)}(\vect(\bA))=\bA\in\mathbb{C}^{N\times M}$.
	\end{itemize}
	\hrulefill
\end{table}

\section{Problem Formulation}\label{sec_ProbForm}
This section is focused on the introduction of the system model accounting for a highly reverberating environment, as well as on the formulation of the constrained optimization problem to jointly design radar transmitter and receiver.
\subsection{System Model}
Let ${\bs}=[s_1,\ldots,s_N]^T\in {\mathbb{C}}^{N}$ be the transmitted fast-time radar code with $N$ being the number of coded sub-pulses.
The observations from the range-azimuth cell under test (CUT) are collected in the vector ${\bv} \in {\mathbb{C}}^{N}$, which can be expressed as \cite{aa}
\begin{equation}
{\bv} = {\alpha _0}{\bs} + {\bc} + {\bn},
\end{equation}
where
\begin{itemize}
	\item 
	$\alpha_0$ accounts for the response of the prospective target within the
	CUT.
	\item
	${\bc}\in{\mathbb{C}}^{N}$ is the signal-dependent interference produced by the range cells adjacent to the CUT, namely
	\begin{equation}
	{\bc} = \sum\limits_{m =  - N + 1 \atop
		m \ne 0}^{N - 1} {{\alpha _m}{{\bJ}_m}{\bs}}, 
	\end{equation}
	with $\alpha_m$ the scattering coefficient associated with the $m$-th range patch.
	 Specifically,	${\left\{ {{\alpha _m}} \right\}_{m \ne 0}}$ are modeled as independent complex, zero-mean, circularly symmetric, random variables with $\E\left[ {{{\left| {{\alpha _m}} \right|}^2}} \right] = {\beta _m},m\ne0$. As a result, the covariance matrix of ${\bc}$ can be cast as
	\begin{equation}
	\begin{split}
	{{\bR}_{\text{d}}^{\left( {\bs} \right)}} = \E\left[ {{\bc}{{\bc}^\dag}} \right] = \sum\limits_{m =  - N + 1\atop m \ne 0}^{m = N - 1} {{\beta_m}{{\bJ}_m}{\bs}{{\bs}^\dag}{\bJ}_m^\dag},
	\end{split}
	\end{equation}
	\item
	${\bn} \in {\mathbb{C}}{^N}$ denotes the signal-independent interference that comprises thermal noise and other disturbances from interfering emitters. It is modeled as a zero-mean, complex, circularly symmetric, random vector with covariance
	$\E\left[ {{\bn}{{\bn}^\dag}} \right] = {{\bR}_{\text{ind}}}$.
\end{itemize}
\subsection{Transmit Code Constraints}\label{constraint}
In this subsection, some constraints on the transmitted signal are introduced to fulfill the appropriate radar requirements \cite{yutvt,7605511}.
\subsubsection{Multi-Spectral Constraints}
To assure spectral compatibility with the surrounding licensed emitters, the radar has to control the spectral shape of the probing waveform to manage the
amount of interfering energy injected on the shared frequency bandwidths. In this respect, let us denote by $K$
the number of licensed emitters while $f^k_1$ and $f^k_2$ indicate the lower and upper normalized frequencies (with respect
to the underlying radar sampling frequency) for the $k$-th
system, respectively. Now, denoting by $S_c\left(f\right)=\left|\sum_{n=1}^{N}s_n e^{-j 2 \pi f n} \right|^2$ the Energy Spectral Density (ESD) of the fast-time code $\bs$, the energy transmitted by the radar within the $k$-th licensed bandwidth (also denoted by stopband) can be expressed as
\begin{eqnarray}\label{eq:vincolo_spettrale}
\int_{f_1^k}^{f_2^k} S_c\left(f\right)df
&=&\bs^{\dag}\bR_I^k\bs,
\end{eqnarray}
where  for $i,l \in \cal N$,
$$\bR_{I}^k(i,l) = (f_2^k-f_1^k) \displaystyle{e^{j  \pi (f_2^k+f_1^k)(i-l)}}\displaystyle{\mbox{sinc}\!\left({\pi (f_2^k\!-\!f_1^k)(i-l)}\right)}.$$
To enable joint radar and communication activities, it is thus demanded 
 that the radar transmitted waveform complies with the constraints
\begin{eqnarray}\label{eq:vincolo_banda}
\bs^\dag\bR_I^k\bs\leq E_I^k,k\in\{1,\ldots,K\},
\end{eqnarray}
where $E_I^k$, $k\in\{1,\ldots,K\}$, accounts for
the acceptable level of disturbance on the $k$-th bandwidth and is tied up to the QoS required by  the $k$-th communication system.

\subsubsection{Constant Envelope and Energy Constraints}
To assure compatibility
	with the amplifier technology and comply with the radar power budget,
constant envelope and energy constraints are forced on the sought waveform, which is tantamount to forcing
\begin{eqnarray}
|s_i|&=&|s_j|,\, i,j\in{\cal{N}},\\
{\left\| {\bs} \right\|^2} &\le& 1.
\end{eqnarray}
\subsubsection{Finite Alphabet Constraint}
As a limited number of bits are available in digital waveform generators, the finite alphabet constraint is possibly forced, requiring $\arg (s_i)\in \Omega_M$, where $\Omega_M=\frac{2\pi}{M}\left\{-\lfloor\frac{M}{2}\rfloor,\ldots,-\lfloor\frac{M}{2}\rfloor+M-1\right\}$ denotes the discrete set of $M$ equi-spaced phases.

\subsubsection{Similarity Constraint}
To bestow some desirable attributes (e.g., Doppler tolerance, ISL, PSL and etc.) to the radar probing signal, a similarity constraint is imposed on the transmitted code, i.e.,
\begin{equation}\label{sm}
{\left\| {\frac{{\bs}}{{\left\| {\bs} \right\|}} - {{\bs}_0}} \right\|_\infty } \le \frac{\epsilon }{{\sqrt N }},
\end{equation}
where $0 \le \epsilon \le 2$ rules the size of the trust hypervolume, and $\bs_0$ is a specific constant modulus reference code with ${\left\| {{{\bs}_0}} \right\|^2} = 1$.

\subsection{Transceiver Design Problem Formulation}
A transceiver design approach aimed at optimizing target detectability under the probing code requirements of Subsection \ref{constraint} is now formalized.
To this end, supposing the received signal $\bv$ filtered via $\bw\in {\mathbb{C}}^{N}$, the SINR at the output of the filter, i.e.,
	\begin{equation}
	\mbox{SINR}(\bs,\bw) = \frac{{{{\left| {{\alpha _0}} \right|}^2}{{\left| {{{\bw}^\dag}{\bs}} \right|}^2}}}{{{{\bw}^\dag}\left( {{\bR}_{\text{d}}^{\left( {\bs} \right)} + {{\bR}_{\text{ind}}}} \right){\bw}}},
	\end{equation}
 is considered as the design metric.

Hence, according to the code limitations introduced  in Subsection \ref{constraint},
the joint transmit-receive pair design assuming either continuous phase or finite alphabet phase codes, respectively, can
 be formulated as the following non-convex and in general NP-hard optimization problems\footnote{If ${\left\{ {{\alpha _m}} \right\}_{m \ne 0}}$ and ${\bn}$ are statistically independent and Gaussian distributed, the SINR optimization is tantamount to maximizing the detection probability.}
\begin{equation}\label{prob_cp}
{\cal{P}}_{\infty} \left \{ \begin{array}{ll}
\displaystyle{\max_{\bs,\bw }} & \overline{\mbox{SINR}}(\bs,\bw )\\
\mbox{s.t.}  & {{\bs}^\dag}{\bR}_I^k{\bs}\le E_I^k, k\in\{1,\ldots,K\},\\
&{\left\| {\bs} \right\|^2} \le 1, \\
&{\left\| {\frac{{\bs}}{{\left\| {\bs} \right\|}} - {{\bs}_0}} \right\|_\infty } \le \frac{\epsilon }{{\sqrt N }},\\
&|s_i|=|s_j|,i,j\in {\cal N},\\
&\arg (s_i)\in \Omega_{\infty},
\end{array} \right.
\end{equation}
and
\begin{equation}\label{prob_dp}
{\cal{P}}_M \left \{ \begin{array}{ll}
\displaystyle{\max_{\bs,\bw  }} & \overline{\mbox{SINR}}(\bs,\bw )\\
\mbox{s.t.}  & {{\bs}^\dag}{\bR}_I^k{\bs}\le E_I^k, k\in\{1,\ldots,K\},\\
&{\left\| {\bs} \right\|^2} \le 1, \\
&{\left\| {\frac{{\bs}}{{\left\| {\bs} \right\|}} - {{\bs}_0}} \right\|_\infty } \le \frac{\epsilon }{{\sqrt N }},\\
&|s_i|=|s_j|,i,j\in {\cal N},\\
&\arg (s_i)\in \Omega_M, 
\end{array} \right.
\end{equation}
with $\overline{\mbox{SINR}}(\bs,\bw) =\mbox{SINR}(\bs,\bw)/{\left| {{\alpha _0}} \right|}^2$ and $\Omega_{\infty}=[-\pi,\pi[$.

\section{Code and Filter Synthesis}\label{CD-based}
In this section, an iterative design procedure is developed  to get optimized transmit-receive pairs leveraging the CD paradigm \cite{CD}. By invoking elementary function, each nonconvex optimization subproblem is solved in closed form, and the monotonic improvement of the SINR is ensured along the iterations.
	As first step toward this goal, let us re-parameterize the optimization vector $\bs$  as
	\begin{equation}
	\bs=\sqrt{P}{\bar \bs}_{\bm \varphi}\odot\bs_0,
	\end{equation}
	where ${\bar \bs}_{\bm \varphi}=[e^{j\varphi_1},\ldots,e^{j\varphi_N}]^T\in\mathbb{C}^N$, with ${\bm \varphi}=\left[\varphi_1,\ldots, \varphi_N\right] ^T\in\mathbb{R}^N$, and $P$ ($0\le P \le 1$) account for code  phases and the signal amplitude level, respectively.
	As a result, the similarity constraint is tantamount to $\Re\{e^{j\varphi_i}\}\geq1-\epsilon^2/2,i\in {\cal{N}}$.
	 Hence, denoting by $\delta=\arccos(1-\epsilon^2/2)$, for any $i\in \cal{N}$
		\begin{itemize}
			\item $\varphi_i\in\Psi_\infty= [-\delta,\delta]$ for the continuous  case\cite{atsp};
		 \item $\varphi_i\in\Psi_M= \frac{2\pi}{M}\left\{{\alpha_\epsilon ,\alpha_\epsilon  + 1, \cdots ,\alpha_\epsilon  + \omega_\epsilon  - 1}\right \}$  with
	$\alpha_\epsilon  =  - \left\lfloor {\frac{{M\delta}}{{2\pi }}} \right\rfloor$ and	
	 \begin{equation*}
	\omega_\epsilon  = \left\{ \begin{array}{l}
	1 -2\alpha_\epsilon, \text{ if $\epsilon  \in \left[ {0,2} \right[$}\\
	M,\;\text{if $\epsilon  = 2$}
	\end{array} \right.
	\end{equation*}
	for the finite alphabet case \cite{9082109}.
\end{itemize}
	Additionally, in the transformed domain, the spectral constraints can be cast as
	\begin{equation}
	P{{{\bar \bs}_{\bm \varphi}}^\dag}\overline{{\bR}}_I^k{{\bar \bs}_{\bm \varphi}}\le E_I^k, k\in\{1,\ldots,K\},
	\end{equation}	
	where $\overline{{\bR}}_I^k=\diag\{\bs_0\}^\dag{{{\bR}}_I^k}\diag\{\bs_0\},k=1,\ldots, K$.	
Finally, the objective function in \eqref{prob_dp} can be expressed as  
\begin{equation}\label{obj}
\begin{split}
\chi\left({\bm \varphi},P,\bw \right)&=\overline{\mbox{SINR}}(\sqrt{P}{\bar \bs}_{\bm \varphi}\odot\bs_0,\bw)\\
&=\frac{P{{{{\bar \bs}_{\bm \varphi}}^\dag}{{\bM}_1^{(\bw)}}{{\bar \bs}_{\bm \varphi}}}}{P{{{{\bar \bs}_{\bm \varphi}}^\dag}{{\bM}_2^{(\bw)}}{{\bar \bs}_{\bm \varphi}} + \vartheta^{(\bw)}}},
\end{split} 
\end{equation}
	with
	\begin{eqnarray*}
	{\bM}_1^{(\bw)}&=&\diag\{\bs_0\}^\dag\bw\bw^{\dag}\diag\{\bs_0\},\\
	{\bM}_2^{(\bw)}&=&\diag\{\bs_0\}^\dag\sum\limits_{m =  - N + 1\atop
		m \ne 0}^{m = N - 1} {{\beta_m}{{\bJ}_m^\dag}{\bw\bw^{\dag}}{\bJ}_m}\diag\{\bs_0\},	\\
	\vartheta^{(\bw)}&=&{{\bw}^{\dag}}{{\bR}_{\text{ind}}}{{\bw}}.
	\end{eqnarray*}

To proceed further, let us introduce the optimization vector $\by=\left[{\bm \varphi}^T,P,\bw^T\right]^T\in{\mathbb{R}}^{N+1}\times{\mathbb{C}}^N$ and denote by $\chi\left(\by \right)$ the function \eqref{obj} evaluated at $([y_1,\ldots,y_N]^T,y_{N+1},[y_{N+2},\ldots,y_{2N+1}]^T)$. Consequently, Problems ${\cal{P}}_{\infty}$ and ${\cal{P}}_M$ can be equivalently cast as
	\begin{equation}\label{recast}
	{\bar{\cal{P}}}_p \left \{ \begin{array}{ll}
	\displaystyle{\max_{\by}} & \chi\left(\by \right)\\
	\mbox{s.t.}  
	&{\bar \bs}_{\bm \varphi}=[e^{jy_1},\ldots,e^{jy_N}]^T,\\
	&y_{N+1}{{{\bar \bs}_{\bm \varphi}}^\dag}\overline{{\bR}}_I^k{{\bar \bs}_{\bm \varphi}}\le E_I^k, k\in\{1,\ldots,K\},\\
	&0\le{y_{N+1}} \le 1,\\
	&y_i\in\Psi_p,i\in {\cal{N}}
	\end{array} \right.
	\end{equation}
where $p$ is either an integer number or $\infty$ and specifies the code alphabet size.

The provided reformulation in \eqref{recast} paves the way for an effective CD-based optimization process, where the design variables of the vector  $\by$ are partitioned into $N+2$ blocks given by $y_1,\ldots,y_{N+1},\by_{N+2}$, with $\by_{N+2}=[y_{N+2},\ldots,y_{2N+1}]^T\in\mathbb{C}^N$ corresponding to the receive filter. 
Hence, at the $h$-th step, $h=1,\ldots,N+2$, of the $n$-th iteration, the $h$-th block is optimized with the other blocks fixed at their previous optimized value. Now, denoting by
$\by^{(n)}=\left[y_1^{(n)},\ldots,y_{N+1}^{(n)},\by_{N+2}^{(n)T}\right]^T$, it follows that
\begin{equation}
\begin{split}
&\chi\left(\by^{(n-1)} \right)\\
%= \overline{\mbox{SINR}}\left(\sqrt{y_{N+1}^{(n-1)}}{\bar \bs}_{\bm \varphi}^{(n-1)}\odot\bs_0,\bw^{(n-1)}\right)\\
\leq& \chi\left(\left[y_1^{(n)},y_2^{(n-1)}\ldots,y_{N+1}^{(n-1)},\by_{N+2}^{(n-1)T}\right]^T\right)\\
\leq&
\cdots\\
\leq&\chi\left(\left[y_1^{(n)},\ldots,y_{N+1}^{(n)},\by_{N+2}^{(n-1)T}\right]^T\right)\\
\leq&\chi\left(\by^{(n)} \right)
\nonumber
\end{split}
\end{equation}

%	\begin{equation}
%	\begin{split}
%	&\chi\left(\by^{(n-1)} \right)= \overline{\mbox{SINR}}\left(\sqrt{y_{N+1}^{(n-1)}}{\bar \bs}_{\bm \varphi}^{(n-1)}\odot\bs_0,\bw^{(n-1)}\right)\\
%	&\leq 
%	\overline{\mbox{SINR}}\left(\sqrt{y_{N+1}^{(n-1)}}{\bar \bs}_{\bm \varphi}^{(n)}\odot\bs_0,\bw^{(n-1)}\right)\\
%	&\le\overline{\mbox{SINR}}\left(\sqrt{y_{N+1}^{(n)}}{\bar \bs}_{\bm \varphi}^{(n)}\odot\bs_0,\bw^{(n-1)}\right)\\ 
%	&\le\overline{\mbox{SINR}}\left(\sqrt{y_{N+1}^{(n)}}{\bar \bs}_{\bm \varphi}^{(n)}\odot\bs_0,\bw^{(n)}\right)=\chi\left(\by^{(n)} \right),\nonumber
%	\end{split}
%	\end{equation}
%	with ${\bar \bs}_{\bm \varphi}^{(n)}=\left[e^{j\varphi_1^{(n)}},\ldots,e^{j\varphi_N^{(n)}}\right]$,
The above inequalities entail that the objective
function values monotonically increase along with the iterations. Thus, being $\chi\left(\by \right)$ upper bounded by $ {1}/{\sigma_0}$ with $\sigma_0$ the variance of thermal noise, $\chi\left(\by^{(n)} \right)$ converges to a finite value.
In the following, the procedures devised to optimize the $h$-th block, $h=1,\ldots,N+2$, at the $n$-th iteration are developed. This represents the main technical achievement of this paper from an optimization theory standpoint.
\subsection{Code Phase Optimization}
Assuming $h\in\cal N$, the block to optimize is $y_h$, with the other blocks set to their previous optimized values, i.e.,
${\by}_{-h}^{(n)}=\left[y_1^{(n)},\ldots,y_{h-1}^{(n)},y_{h+1}^{(n-1)},\ldots,y_{N+1}^{(n-1)},\by_{N+2}^{(n-1)T}\right]^T$, $h\in\cal N$. Hence, as shown in Appendix \ref{A}, the problem to solve is
\begin{equation}\label{phase}
{\cal{P}}_p^{y^{(n)}_h}\left\{ \begin{array}{ll}
\displaystyle{\max_{y_h}} &  \chi\left(y_h;{\by}_{-h}^{(n)}\right)\\
\mbox{s.t.}  & \Re\{z_{k,h}^{(n)}e^{jy_h}\}\le \tilde{c}_{k,h}^{(n)}, k\in\{1,\ldots,K\}\\
& y_h\in \Psi_p
\end{array} \right.
\end{equation}
where    
\begin{equation*}
\begin{split}
\chi\left(y_h;{\by}_{-h}^{(n)}\right)=\frac{\Re\{a_{h}^{(n)}e^{jy_h}\}+b_h^{(n)} }{\Re\{c_h^{(n)}e^{jy_h}\}+d_{h}^{(n)} },h\in{\cal{N}}.
\end{split}
\end{equation*}
The feasible set\footnote{ To avoid unnecessary complications, $\epsilon<2$ is assumed.} $\bar{\mathcal{F}}_p$ of Problem ${\cal{P}}_p^{y^{(n)}_h}$ and the monotonicity properties of $\chi\left(y_h;{\by}_{-h}^{(n)}\right)$ are analyzed in Appendix \ref{proproof} with specification of global maximizer and minimizer $\phi_{g,h}^{(n)}$ and $\phi_{s,h}^{(n)}$, respectively. In particular, it is shown that $\bar{\mathcal{F}}_p$ can be expressed as
\begin{equation}\label{fea}
\bar{\mathcal{F}}_p=\left\{ {\begin{array}{*{20}{ll}}
	\displaystyle{{\mathop  \cup \limits_{i = 1}^{{K_\infty}} }}\left[\hat{l}_i^{\infty},\hat{u}_i^{\infty}\right],&{\text{if } p=\infty},\\
{{\mathop  \cup \limits_{\scriptstyle t = 1}^{K_M} \left\{\hat{l}_t^M,\hat{l}_t^M+\frac{2\pi}{M},\ldots,\hat{u}_t^M\right\}}}&{{\text{if } p=M}},
	\end{array}} \right.
\end{equation}
where (see Appendix \ref{proproof} for details)
\begin{itemize}
	\item 
$K_\infty$ is the number of disjoint intervals that compose $\bar{\mathcal{F}}_\infty$ with $\hat{l}_i^{\infty},\hat{u}_i^{\infty}\in \Psi_\infty$, $i=1,\ldots,K_\infty$, depending on the specific instance of Problem ${\cal{P}}_\infty^{y^{(n)}_h}$ and such that $\hat{l}_i^\infty\le\hat{u}_i^\infty<\hat{l}_{i+1}^\infty,i=1,\ldots,K_\infty-1$.
\item $K_M\le K_\infty$ is the number of disjoint discrete sets that form $\bar{\mathcal{F}}_M$  with $\hat{l}_t^M,\hat{u}_t^M\in\Psi_M,t=1,\ldots,K_M$, and $\hat{l}_t^M\le\hat{u}_t^M<\hat{l}_{t+1}^M, t=1,\ldots,K_M-1$.
\end{itemize}

The following proposition  paves
the way to the solution of the non-convex optimization problem ${\cal{P}}_{\infty}^{y^{(n)}_h}$.

\begin{proposition}\label{pro}
	An optimal solution $y_h^\star$ to ${\cal{P}}_{\infty}^{y^{(n)}_h}$ can be evaluated in closed form via the computation of elementary function as follows:
	\begin{itemize}
		\item if the unconstrained global optimal solution $\phi_{g,h}^{(n)}$ is feasible, i.e., $\phi_{g,h}^{(n)} \in\bar{\mathcal{F}}_{\infty}$, the optimal solution is $x_d^\star=\phi_{g,h}^{(n)}$;
		\item otherwise, if  $\phi_{g,h}^{(n)}\notin  \left[ \hat{l}_1^\infty  ,\hat{u}_{K_\infty}^\infty\right] $, the optimal solution belongs to $\{\hat{l}_{1}^\infty,\hat{u}_{K_\infty}^\infty\}$, i.e.,
		\begin{equation}
		y_h^\star=\arg\max\limits_{y_h\in\{\hat{l}_{1}^\infty,\hat{u}_{K_\infty}^\infty\}}\chi\left(y_h;{\by}_{-h}^{(n)}\right);
		\end{equation}
		\item else, denoting by $r^\star=\max\limits_{\hat{u}_{r}^\infty
			< \phi_{g,h}^{(n)}} r$, the optimal solution is
		\begin{equation}
		\begin{split}
		y_h^\star =
		\arg\max\limits_{y_h\in\{\hat{l}_1^\infty,\hat{u}_{r^\star}^\infty,\hat{l}_{r^\star+1}^\infty,\hat{u}_{K_\infty}^\infty\}}\chi\left(y_h;{\by}_{-h}^{(n)}\right),
		\end{split}
		\end{equation}
	\end{itemize}
\end{proposition}

\begin{proof}
	See Appendix \ref{p1}.
\end{proof}

 According to Proposition \ref{pro}, as long as $\phi_{g,h}^{(n)}$ does not belong to one of the closed and bounded intervals $[\hat{l}_i^{\infty},\hat{u}_i^{\infty}], i=1,\ldots,K_\infty$, the global optimal solution is one of the intervals extremes. As a consequence, embedding $\bar{\mathcal{F}}_M$ to an appropriate union of closed intervals the following corollary holds true.

\begin{corollary}\label{fd}
An optimal solution $y_h^\star$ to ${\cal{P}}_{M}^{y^{(n)}_h}$ can be derived as follows:
\begin{itemize}
	\item if there exists an index $q^\star\in\{1,\ldots,K_M\}$ satisfying $\hat{l}_{q^\star}^M\le\phi_{g,h}^{(n)}\le \hat{u}_{q^\star}^M$, then 
		\begin{equation}
y_h^\star=\arg\max\limits_{y_h\in\left\{\phi^M_l,\phi^M_u\right\}}\chi\left(y_h;{\by}_{-h}^{(n)}\right),
	\end{equation}
	where 
	\begin{eqnarray}
	\phi^M_l=\left\lfloor {\frac{{{\phi_{g,h}^{(n)}}M}}{{2\pi }}} \right\rfloor \frac{{2\pi }}{M},
		\phi^M_u=\left\lceil {\frac{{{\phi_{g,h}^{(n)}}M}}{{2\pi }}} \right\rceil \frac{{2\pi }}{M};
	\end{eqnarray}
	\item otherwise, if $\phi_{g,h}^{(n)}\notin  \left[ \hat{l}_1^M ,\hat{u}_{K_M}^M\right] $, 
	\begin{equation}
	y_h^\star=\arg\max\limits_{y_h\in\{\hat{l}_1^M,\hat{u}_{K_M}^M\}}\chi\left(y_h;{\by}_{-h}^{(n)}\right);
	\end{equation}
	\item else, 
	\begin{equation}
	\begin{split}
	y_h^\star =
	\arg\max\limits_{y_h\in\{\hat{l}_1^M,\hat{u}_{r^\star}^M,\hat{l}_{r^\star+1}^M,\hat{u}_{K_M}^M\}}\chi\left(y_h;{\by}_{-h}^{(n)}\right),
	\end{split}
	\end{equation}
	where $r^\star=\max\limits_{\hat{u}_{r}^M
		< \phi_{g,h}^{(n)}} r$.
\end{itemize}

\end{corollary}
\begin{proof}
	See Appendix \ref{p2}.
\end{proof}

%\begin{corollary} 
%The vector of lower and upper extremes of $\hat{\mathcal{F}}^p$ denoted by ${\hat{\bl}}^p$ and ${\hat{\bu}}^p$, respectively, can be  assigned by the following steps.
%\begin{enumerate}
%	\item Set $\bl=[l_1,\cdots,l_{K_\infty}]$ and $\bu=[u_1,\cdots,u_{K_\infty}]$;
%	\item if $t_{g,h}^{(n)}\in\mathcal{F}_{\infty}$, $\bl:=\text{sort}([\bl,t_{g,h}^{(n)}])$,$\bu:=\text{sort}([\bu,t_{g,h}^{(n)}])$, where sort denotes the sort function;
%	\item 
%	${\hat{\bl}}^c=2\arctan(\bl)$ and ${\hat{\bu}}^c=2\arctan(\bu)$;
%	\item $\bar{\bl}_d=\lceil{\hat{\bl}}^c/\frac{2\pi}{M}\rceil\frac{2\pi}{M}$, and $\bar{\bu}_d=\lfloor{\hat{\bu}}^c/\frac{2\pi}{M}\rfloor\frac{2\pi}{M}$;
%	\item				$\hat{\bl}^M=\bar{\bl}_d(\bar{\bl}_d\le\bar{\bu}_d)$, and $\hat{\bu}^d=\bar{\bu}_d(\bar{\bl}_d\le\bar{\bu}_d)$, where $\ba(\ba\le\bb)$ denotes a vector whose elements satisfy $a_i\le b_i$.
%	
%\end{enumerate}	
%\end{corollary}
Hence, starting from a feasible solution $\by^{(0)}$, the solution to ${\cal{P}}_p^{y^{(n)}_h}$ for different values of the objective parameters can be easily accomplished following the line of \textbf{Algorithm \ref{alg_op}}.  
\begin{algorithm}
	\caption{Phase Code Entry Optimization.}
	\label{alg_op}
	\textbf{Input:} $\by_{-h}^{(n)}$, $p\in\{M,\infty\}$, $h$, $a_h^{(n)}$, $b_h^{(n)}$, $c_h^{(n)}$, $d_h^{(n)}$, $z_{k,h}^{(n)}$, $\tilde{c}_{k,h}^{(n)}$, $k\in\{1,\ldots,K\}$; \\
	\textbf{Output:} $y_h^\star$;
	\begin{enumerate}
\item Compute the feasible set 	$\bar{\mathcal{F}}_{p}$ as given in \eqref{fea};
\item \textbf{If}  $\phi_{g,h}^{(n)} \in \left[ \hat{l}_i^p ,\hat{u}_{i}^p\right]$, $i\in\left\{1,\ldots,K_p\right\}$, $$y_h^\star = \left\{ {\begin{array}{*{20}{c}}
	{\phi _{g,h}^{(n)},}&\text{if }{p = \infty, }\\
\arg\max\limits_{y_h\in\left\{\phi^M_l,\phi^M_u\right\}}\chi\left(y_h;{\by}_{-h}^{(n)}\right)&\text{if }{p = M.}
	\end{array}} \right.$$ 
\item \textbf{Elseif}, $\phi_{g,h}^{(n)} \notin \left[ \hat{l}_1^p ,\hat{u}_{K_p}^p\right]$,
\begin{equation*}
y_h^\star =
\arg\max\limits_{y_h\in\{\hat{l}_1^p,\hat{u}_{K_p}^p\}}\chi\left(y_h;{\by}_{-h}^{(n)}\right).
\end{equation*}
\item \textbf{Else},
	\begin{equation*}
\begin{split}
y_h^\star =
\arg\max\limits_{y_h\in\{\hat{l}_1^p,\hat{u}_{r^\star}^p,\hat{l}_{r^\star+1}^p,\hat{u}_{K_p}^p\}}\chi\left(y_h;{\by}_{-h}^{(n)}\right),
\end{split}
\end{equation*}
where $r^\star=\max\limits_{\hat{u}_{r}^p
	< \phi_{g,h}^{(n)}} r$;
\item \textbf{Output} $y_h^\star$.
	\end{enumerate}
\end{algorithm}

\subsection{Code Amplitude Optimization}
Focusing on $h=N+1$,  ${\bar{\cal{P}}}_p$ w.r.t. ${y_{N+1}}$ reduces to the following problem,
\begin{equation*}
{\cal{P}}_{y_{N+1}^{(n)}}\left\{ \begin{array}{ll}
{\max\limits_{y_{N+1}}} & \frac{{y_{N+1}}{{{{\bar \bs}_{\bm \varphi}}^{(n)\dag}}{{\bM}_1^{\left( \by_{N+2}^{(n-1)}\right)}}{{\bar \bs}_{\bm \varphi}^{(n)}}}}{{y_{N+1}}{{{{\bar \bs}_{\bm \varphi}}^{(n)\dag}}{{\bM}_2^{\left( \by_{N+2}^{(n-1)}\right)}}{{\bar \bs}_{\bm \varphi}}^{(n)} + \vartheta^{\left( \by_{N+2}^{(n-1)}\right)} }}\\
\mbox{s.t.}  &  y_{N+1} \,q_k^{(n)}\le E_I^k, k\in\{1,\ldots,K\}\\
& 0\leq y_{N+1}\leq 1
\end{array} \right.
\end{equation*}
where $q_k^{(n)}={{\bar \bs}_{\bm \varphi}^{(n)\dag}}\overline{{\bR}}_I^k{{\bar \bs}_{\bm \varphi}}^{(n)}$, and ${\bar \bs}_{\bm \varphi}^{(n)}=[e^{j\varphi_1^{(n)}},\ldots,e^{j\varphi_N^{(n)}}]^T$. 
The first-order derivative of the objective function satisfies
\begin{equation*}
\frac{\vartheta^{\left( \by_{N+2}^{(n-1)}\right)} {{{{\bar \bs}_{\bm \varphi}}^{(n)\dag}}{{\bM}_1^{\left( \by_{N+2}^{(n-1)}\right)}}{{\bar \bs}_{\bm \varphi}^{(n)}}}}{\left[ {y_{N+1}}{{{{\bar \bs}_{\bm \varphi}}^{(n)\dag}}{{\bM}_2^{\left( \by_{N+2}^{(n-1)}\right)}}{{\bar \bs}_{\bm \varphi}}^{(n)} + \vartheta^{\left( \by_{N+2}^{(n-1)}\right)} }\right] ^2}>0.
\end{equation*}
Hence, the objective function monotonically increases over $[0,1]$, and the optimal solution is just the highest feasible value, given by
\begin{eqnarray}\label{sn1}
y_{N+1}^{(n)\star}=\displaystyle{\min}\left(\displaystyle{\min_{k\in\{1,\ldots,K\}}}\left({\frac{E_I^k}{q_k^{(n)}}}\right),1\right).
\end{eqnarray}

\subsection{Receive Filter Optimization}
With reference to the $(N+2)$-th variable block, the optimization variable is the  receive filter. The update of the receive filter, i.e., $\by_{N+2}$, is tantamount to solving
	\begin{equation}\label{wsol}
		{\cal{P}}_{\by_{N+2}^{(n)}} \left \{ \begin{array}{ll}
			\displaystyle{\max_{\by_{N+2}}} &\overline{\mbox{SINR}}\left({\bs}^{\left( {n } \right)},\by_{N+2}\right) .
		\end{array} \right.
	\end{equation}
where ${\bs}^{\left( {n } \right)}=\sqrt{y_{N+1}^{\left( {n } \right)}}{\bar \bs}_{\bm \varphi}^{\left( {n } \right)}\odot\bs_0$ is the updated transmit signal at the $n$-th iteration. The optimal solution \cite{aa} is
\begin{equation}\label{w}
{\by_{N+2}^{\left( n \right)\star
}} = \frac{{{{\left( {{\bR}_\text{d}^{\left( {{{\bs}^{\left( {n } \right)}}} \right)} + {{\bR}_{\text{ind}}}} \right)}^{ - 1}}{{\bs}^{\left( {n } \right)}}}}{{{{\bs}^{\left( {n } \right)\dag}}{{\left( {{\bR}_\text{d}^{\left( {{{\bs}^{\left( {n } \right)}}} \right)} + {{\bR}_{\text{ind}}}} \right)}^{ - 1}}{{\bs}^{\left( {n } \right)}}}}.
\end{equation}
\begin{remark}\label{cgm}
Supposing $\bA= {{\bR}_{\rm d}^{\left( {{{\bs}^{\left( {n } \right)}}} \right)} + {{\bR}_{{\rm ind}}}} $ which is positive definite, the solution to the linear equations $\bA{\hat\bw}={\bs}^{\left( {n } \right)}$ can be derived via the Conjugate Gradient Method (CGM) \cite{shewchuk1994introduction}. 
\end{remark}

\subsection{Optimization Process and Computational Complexity}
Starting from a feasible solution $\bs^{(0)}$, the overall iterative design procedure is summarized in
\textbf{Algorithm \ref{alg_1}}.  
Note that in place of the standard cyclic updating rule, the Maximum Block Improvement (MBI) strategy \cite{8454321} can be used, too. As to the computational complexity of \textbf{Algorithm \ref{alg_1}}, it is linear with the number of iterations. In each iteration, it mainly includes the computation of  the code phases (step 3), the code  amplitude (step 4), and the receive filter (step 5), whose complexities are now discussed.

Step 3 includes the solution of ${\cal{P}}_p^{y^{(n)}_h},h\in\cal N$, and the	main actions to perform at each $h$ are: 
\begin{enumerate}
	\renewcommand{\labelenumi}{\theenumi.}
	\item evaluation of the problem parameters;\label{pa}
	\item calculation of $\bar{\mathcal{F}}_p$;\label{fp}
	\item determination of the optimal phase.\label{se}
\end{enumerate}

As to item \ref{pa}, the parameters to compute are   $a_{h}^{(n)}$, $b_{h}^{(n)}$, $c_{h}^{(n)}$, $d_{h}^{(n)}$, ${{z}_{k,h}^{(n)}}$, and $\tilde{c}_{k,h}^{(n)}$, $k\in\{1,\ldots,K\}$, $h\in\cal N$.
Leveraging smart recursive schemes based on suitable support variables, it follows that for a complete algorithm iteration the computational complexity of $a_h^{(n)}$ and $b_h^{(n)}$ is $O(N)$ whereas that of $c_h^{(n)}$ and $d_h^{(n)}$ is $O(N^2\log N)$.	With reference to $z_{k,h}^{(n)}$, it requires $N$ multiplications for any $k\in\{1,\ldots,K\}$ and $h\in\cal N$. Furthermore, following the same line of reasoning as in \cite{atsp}, for any $k\in\{1,\ldots,K\}$, $\tilde{c}_{k,h}^{(n)}$, $h\ge 2$ can be updated by 
 canceling out the items related to the variable $y_{h}^{(n-1)}$ from  $\tilde{c}_{k,h-1}^{(n)}$, and adding those involving $y_{h-1}^{(n)}$, 
 with a  complexity of $O(N)$; similar considerations hold true with respect to $h=1$.
 Hence, the overall
computational complexity of item \ref{pa} at each algorithm iteration is $O(N^2(\log N+K))$.  
As to item \ref{fp}, note that the determination of $\bar{\mathcal{F}}_\infty$ can be performed with a computational complexity of $O(K\log(K))$ for any $h\in\cal N$ leveraging the results of \cite{atsp}.
Besides, according to Appendix \ref{proproof}, $\bar{\mathcal{F}}_M$ can be obtained resorting to appropriate quantizations of $\hat{l}_i^\infty$ and $\hat{u}_i^\infty$, $i=1,\ldots,K_\infty$ with an extra computational complexity of $O(K)$. With reference to item \ref{se}, the unconstrained optimal solution $\phi_{g,h}^{(n)}$, for any $h\in\cal N$, can be evaluated via elementary functions using $a_{h}^{(n)}$, $b_{h}^{(n)}$, $c_{h}^{(n)}$, and $d_{h}^{(n)}$ with a computational burden of $O(1)$, while the optimal solution to the Problem ${\cal{P}}_p^{y^{(n)}_h}$ can be accomplished with a complexity at most of $O(K)$ positioning $\phi_{g,h}^{(n)}$ within $\{\hat{l}_1^p,\hat{u}_1^p,\ldots,\hat{l}_{K_p}^p,\hat{u}_{K_p}^p\}$, $p\in\{\infty,M\}$ in the correct sorted location.

With reference to the complexities of step 4 and 5, the former (mainly related to the efficient computation of $q_k^{(n)}$) is $O(K)$, the latter requires $O(N^3)$ operations, with the most demanding task given by the evaluation of ${\bR}_\text{d}^{\left(\!{{{\bs}^{\left( {n } \right)}}}\! \right)}$; in particular, 
 according to the Remark \ref{cgm}, \eqref{w} can be efficiently computed by CGM with a complexity of $O(N^2Q)$  where $Q$ is the number of iterations.  Therefore, the overall computational complexity of steps 3, 4 and 5 is  $O(N^2(K+N)+KN\log K)$. 
%\footnote{{\color{blue}It can be shown that in a homogeneous clutter environment, i.e., $\beta_i=\beta_j$, $i,j=\pm 1,\ldots,\pm(N-1)$, the overall complexity of steps 3, 4 and 5 can be reduced to $O(N^2(K+Q)+KN\log K)$, exploiting the Toeplitz structures of the matrices involved in the parameters definition. }}
\begin{algorithm}
	\caption{Transceiver design method.}
	\label{alg_1}
	\textbf{Input:} Reference code $\bs_0$, phase cardinality $p$, initial feasible code $\bs^{(0)}\in\Omega_\infty$ (resp. $\Omega_M$),  $\epsilon$, minimum required improvement $\bar{\epsilon}_1$, $\beta_m$, $\bR_{\text{ind}}$,  $f_1^k$, $f_2^k$ and $E_I^k$, $k\in\{1,\ldots,K\}$;  \\
	\textbf{Output:} Optimized solution $\by^\star$;
	\begin{enumerate}
		\item {\bf Initialization}.
	\begin{itemize}
		%\item Generate .
		\item Set  $n := 0$;
		\item Compute $\by_{N+2}^{(0)}$ by \eqref{w};
		\item Set $y_h^{(0)}=\arg\left( s_h^{(0)} s_{0,h}^{*}\right),h\in\mathcal{N}$;
		\item  $\by^{(0)}=\left[y_1^{(0)} ,\ldots,y_N^{(0)},\left\| \bs^{(0)}\right\| ^2,\by_{N+2}^{(0)T}\right]^T$;
		\item Compute $\chi\left(\by^{(0)} \right)$;
	\end{itemize}
	\item  Set ${n := n+1}$;\label{2}
	
	\item	
	\textbf{For} $h= 1:N$ 
	\begin{itemize}
		\item Update $y_h^{(n)}$ by \textbf{Algorithm \ref{alg_op}};
	\end{itemize}
	\textbf{End} 

	\item
	Update $y_{N+1}^{(n)}$ by \eqref{sn1};
	% 	    \item
	% 	    $\bs^{(n)}=\sqrt{y_{N+1}^{(n)}}\bx^{(n)}\odot\bs_0$;
	
	\item Update $\by_{N+2}^{(n)}$ by \eqref{w};\label{4}    
	\item  $\by^{(n)}=\left[y_1^{(n)},\ldots,y_{N+1}^{(n)},\by_{N+2}^{(n)T}\right]^T$;
	\item Compute $\chi\left(\by^{(n)} \right) $. 
	\item {If} $\left|\chi\left(\by^{(n)} \right) - \chi\left(\by^{(n-1)} \right)\right| \ \le \bar{\epsilon}_1$, stop. Otherwise, go to the step \ref{2};
	\item \textbf{Output} $\by^\star=\by^{(n)}$.

	\end{enumerate}
\end{algorithm}

\section{Heuristic Methods for Algorithm Initialization}
The solution provided by \textbf{Algorithm \ref{alg_1}} depends on the initial feasible sequence $\bs^{(0)}$. Thus, the development of a heuristic procedure ensuring high quality starting points is valuable.
To this end, an ad-hoc transceiver synthesis strategy is proposed, accounting for the spectral constraints via a penalty term in the objective function. Specifically, the following design problem is considered
\begin{equation} \label{fe_1}
 \left\{ \begin{array}{ll}
\displaystyle{ \max _{\bs,\bw} } &f(\bs,\bw)\\
\mbox{s.t.}  
%&{\left\| {\bs} \right\|^2} = N, \\
&{\left\| {{{\bs}} - {{\bs}_0}} \right\|_\infty } \le \frac{\epsilon }{{\sqrt N }}\\
&|s_i|=1/\sqrt{N},i\in {\cal N}\\
&\arg (s_i)\in \Omega_p
\end{array} \right. 
\end{equation}
with 
\begin{equation}\label{f}
f(\bs,\bw)=\overline{\mbox{SINR}}(\bs,\bw) -\bar\beta\bs^\dag\bR\bs
\end{equation}
where $\bar\beta\geq 0$  is a weighting factor ruling the relative importance between the two objectives\footnote{$\bar{\beta}$ can be interpreted  as a weight which scalarizes  the multi-objective optimization problem involving the SINR and the opposite total interference energy on the licensed bands as figure of merits.}, and $\bR=\sum\limits_{k = 1}^{K} \frac{{\bR}_I^k}{E_I^k}$. 
 Thus, denoting by $\bs_1^{\star}$ an optimized solution to \eqref{fe_1}, the initial feasible code to {\bf Algorithm \ref{alg_1}}  can be constructed as
\begin{equation}
\bs_2=\frac{\bs_1^{\star}}{\sqrt{\max\left(1,\displaystyle{\max_{k\in\{1,\ldots,K\}}\bs_1^{\star \dag} \frac{{\bR}_I^k}{E_I^k}\bs_1^{\star}}\right)}}.
\end{equation}

Note that Problem \eqref{fe_1} is in general NP-hard. As a first step to handle this synthesis, let us re-parameterize  the transmit signal as $\bs=\left[e^{j\phi_1},\ldots, e^{j\phi_N}\right]^T\odot\bs_0 $; hence, Problem \eqref{fe_1} boils down to 
\begin{equation} \label{fe_2}
\left\{ \begin{array}{ll}
\displaystyle{ \max _{\bq} } &\hat{f}(\bq)\\
\mbox{s.t.}  
&\bq \in {\cal{C}}_p
\end{array} \right. 
\end{equation}
where $\bq=[\bar{\bq}^T,\widetilde{\bq}^T]^T\in{\mathbb{C}}^{2N}$ is the new optimization vector (with $\bar{\bq}=[e^{j\phi_1},\ldots,e^{j\phi_N}]^T$ and $\widetilde{\bq}=\bw$), $ {\cal{C}}_p={\cal{D}}_p^N\times{\mathbb{C}}^N\subseteq{\mathbb{C}}^{2N}$ (with ${\cal{D}}_p=\{x=e^{j\phi}|\phi\in\Psi_p\}$) is the feasible set of  $\bq$, and $\hat{f}(\bq)=f\left({\bar{\bq}}\odot\bs_0,\widetilde{\bq}\right)$, namely it denotes the objective in \eqref{f} evaluated at $\left(\bs,\bw \right)=\left({\bar{\bq}}\odot\bs_0,\widetilde{\bq}\right)$.

To solve Problem \eqref{fe_2}, the optimization framework in \cite{7366709} and \cite{8454321} is used. The main idea of this solution strategy is to partition the variables of the optimization vector $\bq$ into $I$ decoupled blocks\footnote{It is assumed that ${\cal{C}}_p={\cal{C}}_{p,1}\times\ldots\times{\cal{C}}_{p,I}$, with ${\cal{C}}_{p,i}$ the feasible set of $i$-th block of variables, $i=1,\ldots,I$.}, i.e., avoiding useless complications, $\bq=[\bq_1^T,\ldots,\bq_I^T]^T$, where $\bq_i\in{\cal{C}}_{p,i}\subseteq{\mathbb{C}}^{Q_i}$ with $\sum_{i=1}^{I}Q_i=2N$, and update one block at a time\footnote{Either an alternating or an MBI strategy can be considered.}. The procedure is reported in \textbf{Algorithm \ref{alg_2}} assuming an alternating 
optimization rule and no approximations to the feasible set. Therein, focusing on the optimization of the $i$-th variables block ($i=1,\ldots,I$), $\tilde{f}_i(\bx_i;\bq)$, for $\bq\in{\cal{C}}_p$ and $\forall \bp$, represents the surrogate function to maximize, which shares the following properties \cite{7366709,8454321}:
\begin{enumerate}
\renewcommand{\labelenumi}{(P\theenumi)}
\item $\tilde{f}_i(\bx_i;\bq)$ is continuous with respect to $\bx_i$ and $\bq$;
\item $\tilde{f}_i(\bx_i;\bq)\leq \hat{f}([\bq_1^T,\ldots,\bq_{i-1}^T,\bx_i^T,\bq_{i+1}^T,\ldots,\bq_I^T]^T)$  for $\bx_i\in{\cal{C}}_{p,i}$ and $\forall \bq\in{\cal{C}}_{p}$;
\item $\tilde{f}_i(\bq_i; [\bq_1^T,\ldots,\bq_{i-1}^T,\bq_i^T,\bq_{i+1}^T,\ldots, \bq_I^T]^T) =  \hat{f}([\bq_1^T,\ldots,\bq_{i-1}^T,\bq_i^T,\bq_{i+1}^T,\ldots, \bq_I^T]^T), \forall \bq_i\in{\cal{C}}_{p,i}, \forall \bq\in{\cal{C}}_{p}$.
\end{enumerate}
According to \cite[Proposition 2]{8454321} and \cite[Theorem 1]{7366709}, \textbf{Algorithm \ref{alg_2}} ensures a monotonic improvement of the objective, and, under some mild technical conditions, the convergence to a stationary point for any limit point of the generated sequence of the feasible solutions.

\begin{algorithm}
	\caption{Alternating Optimization for Initialization.}
	\label{alg_2}
	\textbf{Input:} A feasible starting point 
	$\bq^{(0)}=[\bq^{(0)T}_1,\cdots,\bq^{(0)T}_I]^T$, $p\in\{\infty,M\},{\cal{C}}_{p,i},i=1,\ldots,I$.  \\
	\textbf{Output:} Optimized solution $\bs_1^\star$ to Problem \eqref{fe_1};
	\begin{enumerate}
		\item Set $l=0$;
		\item  \textbf{Repeat};
		\item	
		\textbf{For} $i= 1:I$,
	\item	$\bq^{(l+1)\star}_i=\arg \displaystyle{ \max _{\bx_i\in {\cal{C}}_{p,i}} } \tilde{f}_i(\bx_i;\bq^{(l)});
$
			\item $\bq^{(l+1)}=
	\left[ \bq_1^{(l) T},\cdots,\bq_{i-1}^{(l) T},\bq^{(l+1)\star T}_i,\bq_{i+1}^{(l) T},\cdots,\bq^{(l)T}_I\right] ^T$;
	\item  ${l := l+1}$;
	\item	\textbf{End} 
		\item \textbf{Until convergence};
	\item \textbf{Output} $\bs_1^\star=[q^{(l)}_1,\ldots,q^{(l)}_N]^T\odot\bs_0$.
	\end{enumerate}
\end{algorithm}

 In the following subsections, two solution techniques leveraging the optimization framework in \textbf{Algorithm \ref{alg_2}} are proposed to tackle \eqref{fe_2}. The former considers as optimization blocks  $\bar{\bq}$ and $\widetilde{\bq}$. The latter assumes $N+1$ optimization blocks, given by ${q}_1,\ldots,{q}_N,\widetilde{\bq}$.

\subsection{Heuristic  Initialization Via Alternating Optimization with MM (HIVAM)}
The optimization variable $\bq$ is partitioned into two blocks, i.e., $\bq_1=\bar{\bq}\in{\mathbb{C}}^{N}$, $\bq_2=\widetilde{\bq}\in{\mathbb{C}}^{N}$.
At $(l+1)$-th step, the surrogate functions involved in the optimization of  $\bq_1$ and $\bq_2$ are given, respectively, by (see Appendix E in the supplemental material for details) 
%(see the supplemental material for details) 
\begin{equation}
\tilde{f}_1(\bx_1;\bq^{(l)})=\Re\left\{\bz^{({\bq}_1^{(l)},{\bq}_2^{(l)})\dag}\bx_1\right\}+r^{({\bq}_1^{(l)},{\bq}_2^{(l)})},
\end{equation}
\begin{equation}\label{ww}
\begin{split}
\tilde{f}_2(\bx_2;\bq^{(l)})&= \hat{f}([\bq_1^{(l)T},\bx_2^T]^T)\\&=\overline{\mbox{SINR}}\left(\bs^{(l)},{\bx_2}\right) -\bar\beta\bs^{(l)\dag}\bR\bs^{(l)},
\end{split}
\end{equation}
where $\bs^{(l)}={{\bq}_1^{(l)}}\odot\bs_0$.
Moreover, ${\cal{C}}_{p,1}={\cal{D}}_p^N$, and ${\cal{C}}_{p,2}={\mathbb{C}}^N$.

As a consequence, at $(l+1)$-th step,
\begin{enumerate}
\item  the first block is updated (the interested reader may refer to Appendix \ref{33} for technical details) as
\begin{equation}\label{q1}
q_{1,i}^{(l+1)\star}=e^{j\phi_i^{(l)\star}}.
\end{equation}
where 
\begin{itemize}
\item 
if $p=\infty$, $
\phi _i^{(l)\star} = \max \left( {\min \left( {\arg \left( { z_{l,i}} \right),\delta } \right), - \delta } \right);$
\item  
if $p=M$,
$$
	\phi _i^{(l)\star}  = \frac{2\pi}{M}
\max\left( \min\left(  \lceil{\frac{\arg(z_{l,i})M}{2\pi} } \rfloor,m_{\epsilon,M}\right) ,\alpha_\epsilon\right),
$$ with $m_{\epsilon,M}=\left\{\begin{array}{l}M / 2, \text{ if $M$ is even, and $\epsilon=2$}\\ \alpha_{\epsilon}+\omega_{\epsilon}-1,\text{else}\end{array}\right.$.
\end{itemize}
\item the second block is updated as
\begin{equation}\label{q}
{{\bq}_2^{\left( l+1 \right)\star
}} = \frac{{{{\left( {{{\bR}_\text{d}}\left( {\bs^{(l)}} \right) + {{\bR}_{\text{ind}}}} \right)}^{ - 1}}\bs^{(l)}}}{{{\bs^{(l)\dag}}{{\left( {{{\bR}_\text{d}}\left( {{\bs^{(l)}}} \right) + {{\bR}_{\text{ind}}}} \right)}^{ - 1}}{{\bs}^{\left( {l} \right)}}}}.
\end{equation} 
\end{enumerate}

%The description of  HIVAM is reported in  \textbf{Algorithm \ref{ao}}.
%\begin{algorithm}
%	\caption{HIVAM Procedure.}\label{ao}
%	\textbf{Input:}  $\bs_0$, $\epsilon$, $\beta_m$, $k=\pm 1,\ldots,\pm (N-1)$,  $\bR_{\text{ind}}$,  $\Omega_k$ and $E_I^k$, $k\in\{1,\ldots,K\}$, $\bar\beta$, $\bar{\epsilon}_2$;\\
%	\textbf{Output:} $\bs_1$;
%	\begin{enumerate}
%		\item  Set $l := 0$, ${\bq}_1^{(0)}=\bone_N$, and update $ {\bq}_2^{(0)}$ by \eqref{q}; 		
%	{\color{red} 	\item $l:=l+1$;\label{n+1}
%		\item  
% Compute $\bz_{l}$ by \eqref{zl};
%		\item Update $q_{1,i}^{(l)\star}=e^{j\phi_i^{(l-1)\star}}$, where $\phi_i^{(l-1)\star}$ is updated by either \eqref{phi1} or \eqref{phi2}.
%	
%\item
%$\bq_1^{(l)}=[q_{1,i}^{(l)\star},\ldots,q_{1,N}^{(l)\star}]^T$, $\bs^{(l)}={{\bq}_1^{(l)}}\odot\bs_0$;
%\item $\bq^{(l)}=
%\left[ \bq_1^{(l)T},\bq^{(l-1)T}_2\right] ^T$; 
%		\item $l:=l+1$, and update ${\bq}_2^{(l)}$ by \eqref{q};
%			\item $\bq^{(l)}=
%		\left[ \bq_1^{(l-1)T},\bq^{(l)T}_2\right] ^T$;
%		\item {If} $\left|\hat f(\bq^{(l)}) - \hat f(\bq^{(l-2)})\right|  \le \bar{\epsilon}_2$,} output ${\bs_1}=\bs^{(l)}$. Otherwise, go to the step \ref{n+1};
%
%	\end{enumerate}
%\end{algorithm}
\subsection{ Heuristic Initialization via Alternating Optimization with CD (HIVAC)}
 The optimization variable $\bq$ is partitioned into $N+1$ blocks, i.e., $q_i=\bar{q}_i,i\in\cal N$, $\bq_{N+1}=\widetilde{\bq}\in{\mathbb{C}}^{N}$.
	The surrogate functions are obtained restricting the objective function to $\bq_i,i=1,\ldots,N+1$, namely,
	\begin{equation*}
 \tilde{f}_i(\bx_i;\bq^{(l)})=\hat{f}([\bq^{(l)T}_{1},\ldots,\bq^{(l)T}_{i-1},\bx_i^T,\bq^{(l)T}_{i+1},\ldots,\bq^{(l)T}_{N+1}]^T).
\end{equation*}
Otherwise stated,
\begin{equation}\label{fq}
\tilde f_i\left( x_{i};\bq^{(l)}\right) =\frac{\Re\{a_s^{(l)}x_{i}\}+b_s^{(l)}}{\Re\{c_s^{(l)}x_{i}\}+d_s^{(l)}}+\Re\{f_s^{(l)}x_{i}\}+g_s^{(l)},i\in\cal N
\end{equation}
\begin{equation}\label{ww2}
\tilde{f}_i(\bx_{N+1};\bq^{(l)})=\overline{\mbox{SINR}}\left(\bs^{(l)},{\bx}_{N+1}\right) -\bar\beta\bs^{(l)\dag}\bR\bs^{(l)},
\end{equation}
with   $a_s^{(l)}$, $b_s^{(l)}$, $c_s^{(l)}$, $d_s^{(l)}$, $f_s^{(l)}$, $g_s^{(l)}$ reported in Appendix \ref{hivac}. 
Besides,  ${\cal{C}}_{p,1}=\ldots={\cal{C}}_{p,N}={\cal{D}}_p$, and ${\cal{C}}_{p,N+1}={\mathbb{C}}^N$. 

As a consequence, at $(l+1)$-th step,
	\begin{enumerate}
	\item $q_i,i\in\cal N$ can be updated (see  Appendix \ref{hivac} for details) as
  \begin{equation}
q_i^{(l+1)\star}=e^{j\phi_i^\star},
\end{equation}
with
\begin{itemize}
	\item  
$\phi_i^\star=\arg\max\limits_{\phi\in\bar{\mathcal{T}}_{\infty,i}^{(l)}}\tilde {f}_i(e^{j\phi};\bq^{(l)})$, if $p=\infty$, 
with  $\bar{\mathcal{T}}_{\infty,i}^{(l)}$  the set of at most eight points  defined in \eqref{ma} of the supplemental material;
	\item 
$\phi_i^\star=\arg\max\limits_{\phi\in\bar{\mathcal{T}}_{M,i}^{(l)}}\tilde {f}_i(e^{j\phi};\bq^{(l)})$, if $p=M$, 
where $\bar{\mathcal{T}}_{M,i}^{(l)}$ is the set  of at most fourteen points specified in \eqref{ma2} of the supplemental material;
\end{itemize}
\item the $(N+1)$-th block is optimized as
\begin{equation}\label{q2}
{{\bq}_{N+1}^{\left( l+1\right)\star
}} = \frac{{{{\left( {{{\bR}_\text{d}}\left( {\bs^{(l)}} \right) + {{\bR}_{\text{ind}}}} \right)}^{ - 1}}\bs^{(l)}}}{{{\bs^{(l)\dag}}{{\left( {{{\bR}_\text{d}}\left( {{\bs^{(l)}}} \right) + {{\bR}_{\text{ind}}}} \right)}^{ - 1}}{{\bs}^{\left( {l } \right)}}}}.
\end{equation} 
where ${\bs^{(l)}}=\left[q_1^{(l)},\ldots, q_N^{(l)}\right] ^T\odot\bs_0$.
\end{enumerate}

\section{Performance Analysis}\label{sec_performance}
This section is devoted to the performance assessment of {\bf Algorithm \ref{alg_1}} 
in terms of achievable target detectability, spectral shape
of the synthesized transmit waveform and receive filter, and cross-correlation features.  In this respect, a radar with a two-sided bandwidth of 2 MHz and a pulse length of 100 $\mu$s (leading to $N=200$) is considered. As to the reference code $\bs_0$, a unitary energy Linear Frequency Modulated (LFM) pulse  of 100 $\mu$s and a chirp rate  $K_s = (1950 \times 10^3)/(100 \times 10^{-6})$ Hz/s is employed for the continuous  phase. The $M$-quantized version of the above chirp is instead considered for the finite alphabet case, with cardinality $M${\footnote{Unless otherwise stated,  three different initializations are adopted for \textbf{{ Algorithm \ref{alg_1}}} (with $\bar \epsilon_1=10^{-4}$): HIVAM (with $\bar \epsilon_2=10^{-2}$), HIVAC (with $\bar \epsilon_3=10^{-2}$), and the third corresponds to the optimized code at the previous $\epsilon$ value. The one providing the highest $\overline{\mbox{SINR}}$ among the three synthesized sequences is picked up. As to the hyperparameter $\bar\beta$ involved in the initialization process, it has been empirically shown that $\bar\beta_1=1.8675$ and $\bar\beta_2=0.0093$ provide satisfactory performance in all the scenarios and can be used in HIVAM  and HIVAC, respectively. }}.

The covariance matrix of signal-independent interference is
\begin{eqnarray}
{{\bR}_{\text{ind}}} = \sigma_0 \bm{I} + \sum_{k=1}^{2} \frac{\sigma_{I,k}}{\Delta f_k} \bm{R}_I^k + \sum_{k=1}^{2} \sigma_{J,k}\bm{R}_{J,k},
\end{eqnarray} where $\sigma_0=0$ dB is the thermal noise level; 
$\sigma_{I,k}=10$ dB, $k=1,2$, accounts for  the energy of the $k$-th licensed
emitter operating over the normalized frequency interval $\Omega_k=[f_1^k,f_2^k]$ with $\Delta f_k = f_2^k-f_1^k$ the related bandwidth extent ($\Omega_1=[0.2112,0.2534]$, $\Omega_2=[0.5856,0.6112]$); $\bR_{J,k},k=1,2$ is the normalized
	covariance matrix of the $k$-th non-licensed active source, whose normalized carrier frequency, bandwidth, and power are denoted by 
	$f_{J,k}$, $\Delta_{J,k}$ and $\sigma_{J,k\,\text{dB}}$, respectively
	($\sigma_{J,1\,\text{dB}} =35$ dB,  $f_{J,1}= 0.823$, $\sigma_{J,2\,\text{dB}} =40$ dB,  $f_{J,2} = 0.925$, $\Delta_{J,k}=0.001$,$k=1,2$).  For the signal-dependent interference, $\beta_m=8$ dB, $m\in\{\pm1,\ldots,\pm( N-1)\}$, is assumed.
Besides, the radar probing waveform is required to
fulfill the spectral compatibility constraints corresponding to ${E_I^1}_{\text{dB}} = -30$ dB and ${E_I^2}_{\text{dB}} = -35$ dB, respectively. 

\subsection{Target Detectability Assessment}

Fig. \ref{sinr} depicts the normalized SINR achieved by the proposed transceiver design strategy versus $\epsilon$ for continuous and discrete phase codes ($M=2,4,8,16,32,64$).  As expected, a larger similarity parameter leads to a higher  $\overline{\mbox{SINR}}$ value regardless of phase cardinality being available more degrees of freedom (DOF) at the design stage. 
Besides, the finer the phase discretization, the better the performance, with $\overline{\mbox{SINR}}$ curves of the synthesized discrete phase codes closer and closer to that of the continuous phase benchmark. Finally, it is worth pointing out that  the designed sequence coincides with a scaled version of reference code, with an
energy modulation implemented to comply with the forced
spectral constraints if $\epsilon=0$ (for the continuous phase) or $\alpha_\epsilon=0$ (for the finite alphabet).  For instance,  for the binary case  $\alpha_\epsilon=0$ as long as $\epsilon\neq2$, and thus a $\overline{\mbox{SINR}}$ performance improvement just appears at  $\epsilon=2$.
	Consequently,  the similarity parameter should be carefully selected to balance the detection performance and waveform characteristics.
	\begin{figure}
	\centering
		\includegraphics[width=0.95\columnwidth]{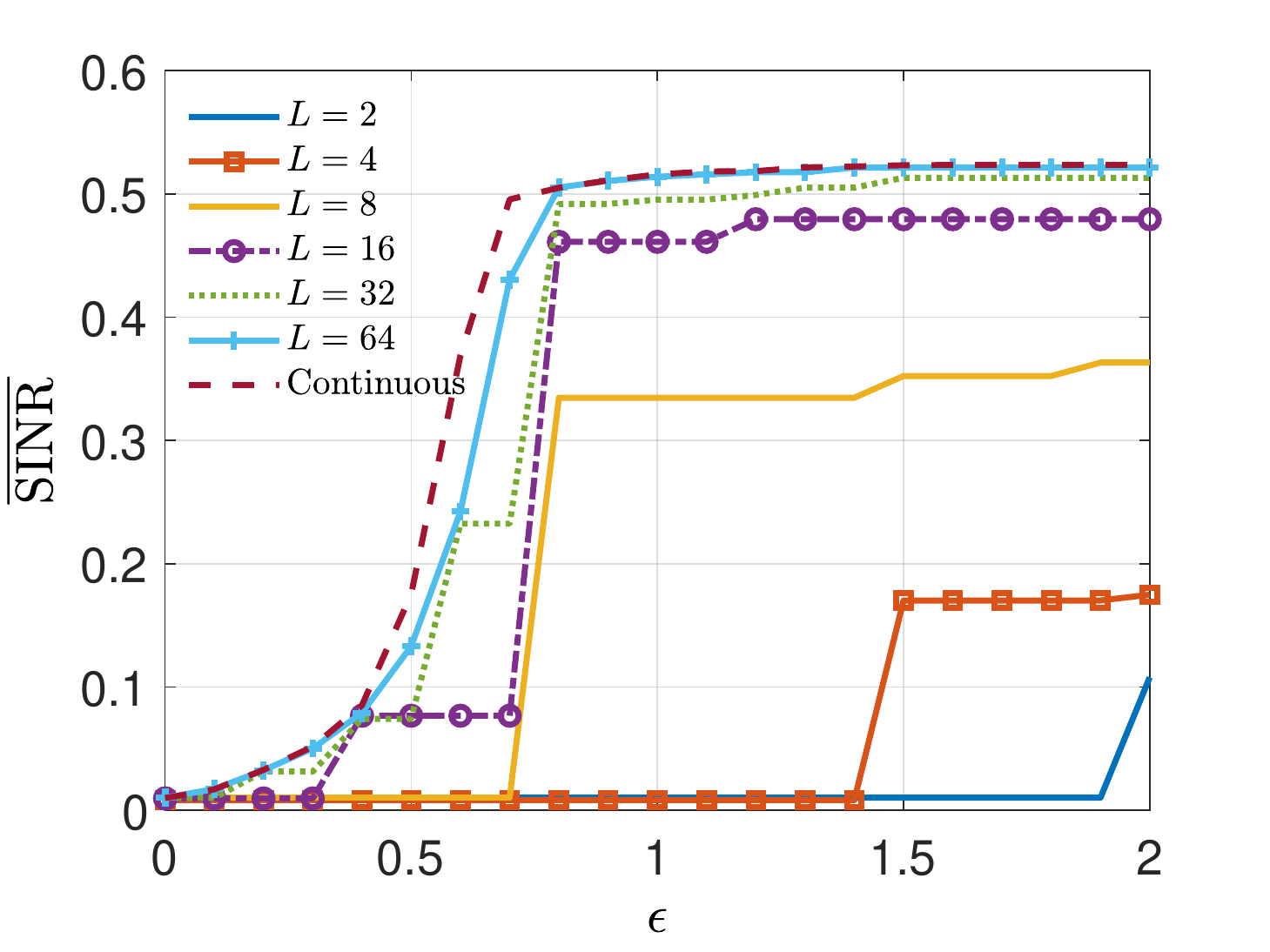}

	\caption{Achieved $\overline{\mbox{SINR}}$ versus  $\epsilon$ for  the continuous and discrete phase codes.}\label{sinr}
\end{figure}

%To further assess the impact of the similarity  parameter and cardinality of  finite alphabet, 
Now, assuming ${\left\{ {{\alpha _m}} \right\}_{m \ne 0}}$ and ${\bn}$ statistically independent and Gaussian distributed as well as a Swerling 0 target, the probability of detection ($P_d$) versus $|\alpha_0|^2$, as function of $\epsilon$ and the phase discretization step, is shown in Fig. \ref{fig:pd}. Therein, the same environmental characterization as in Fig. \ref{sinr} is considered and the false alarm probability ($P_{fa}$) is set to $10^{-4}$. As expected, increasing of  similarity parameter, the cardinality of the code
alphabet, and $|\alpha_0|^2$  provide $P_d$ improvements.
	\begin{figure}
	\centering
	\includegraphics[width=0.95\columnwidth]{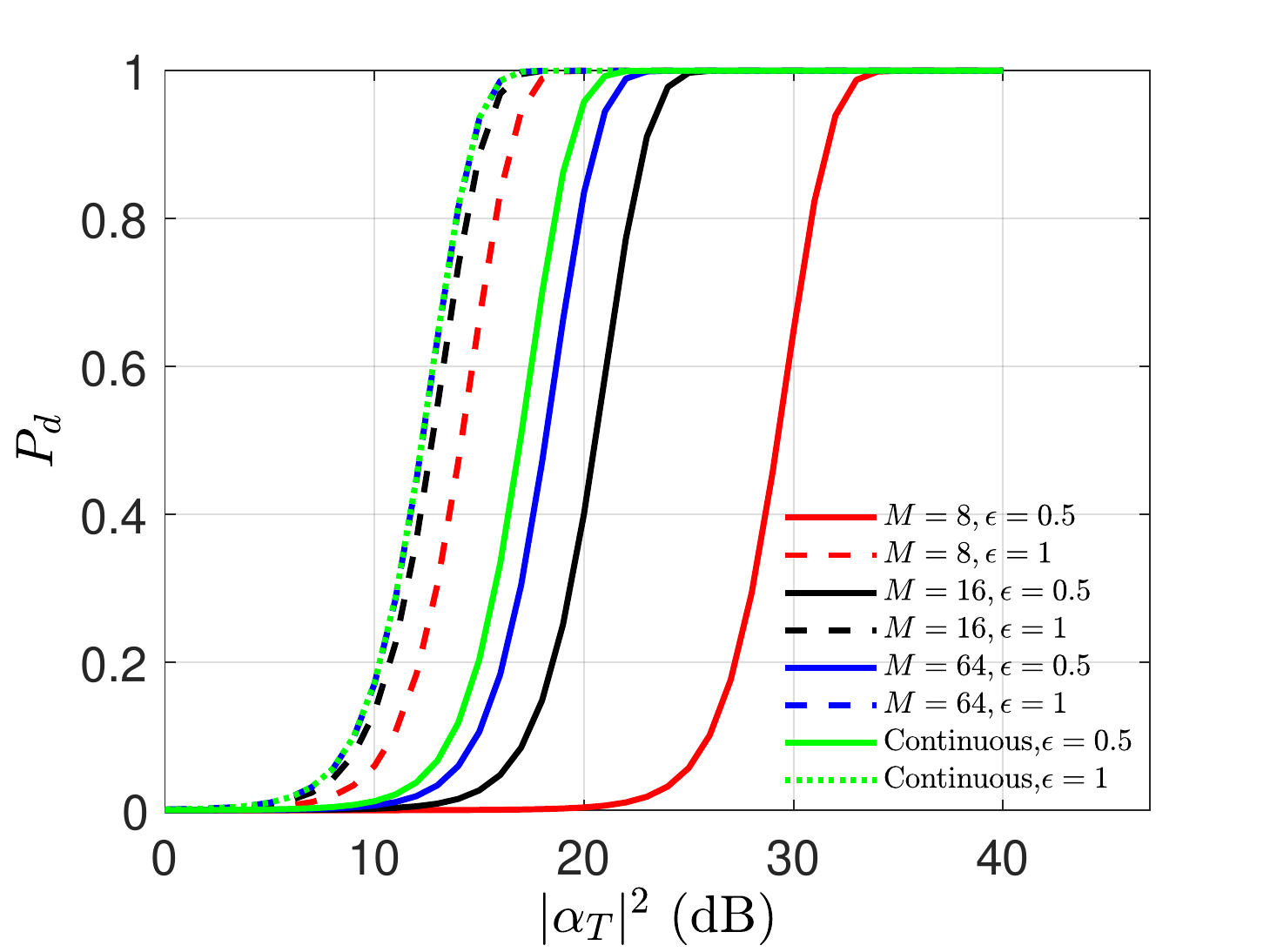}
	
	\caption{$P_d$ versus  $|\alpha_0|^2$ (in dB) for different values of $\epsilon$ and phase code cardinality, assuming $P_{fa}=10^{-4}$.}\label{fig:pd}
\end{figure}

To the best of the Authors' knowledge, the  waveform design methods currently available in the open literature are not able to solve Problem ${\cal{P}}_p,p\in\{\infty,M\}$ in its general form. Aimed at providing a complete assessment of {\bf Algorithm \ref{alg_1}}, in Subsection \ref{com} some specific instances of Problem ${\cal{P}}_p,p\in\{\infty,M\}$ are analyzed, reporting the comparison of our new method with other suitable approaches already devised in the open literature \cite{b1,b2}.

\subsection{Transceiver Characteristics}
In Fig. \ref{fig:esd}, the spectral behavior of the signals synthesized through \textbf{Algorithm \ref{alg_1}}
 (in terms of ESD versus the normalized frequency) is provided as function of the similarity parameter. Specifically,
Fig.  \ref{fig:esd} (a) refers to the continuous phase design, whereas Fig.  \ref{fig:esd} (b) is related to $M=64$. Therein, the stopbands $\Omega_i,i=1,2$ are shaded in light gray. The curves highlight the capability of the devised techniques to suitably control the amount of interference energy produced over
the shared frequency bandwidths, as required by the imposed
 spectral compatibility constraints. Furthermore,
inspection of the figures reveals that, regardless of the code cardinality, the ESD curve corresponding to $\epsilon=0.001$ almost coincides with a scaled version of the considered reference code, as a result of the
energy modulation performed to ensure cohabitation, along with the similarity requirement. Finally, an improvement in the ``useful'' energy
distribution is achieved as $\epsilon$ increases, with deeper and deeper spectral notches in correspondence of the jammed frequencies, especially for the continuous alphabet code.
\begin{figure}
	\centering
	\subfigure[Continuous phase]{
		\includegraphics[width=0.95\columnwidth]{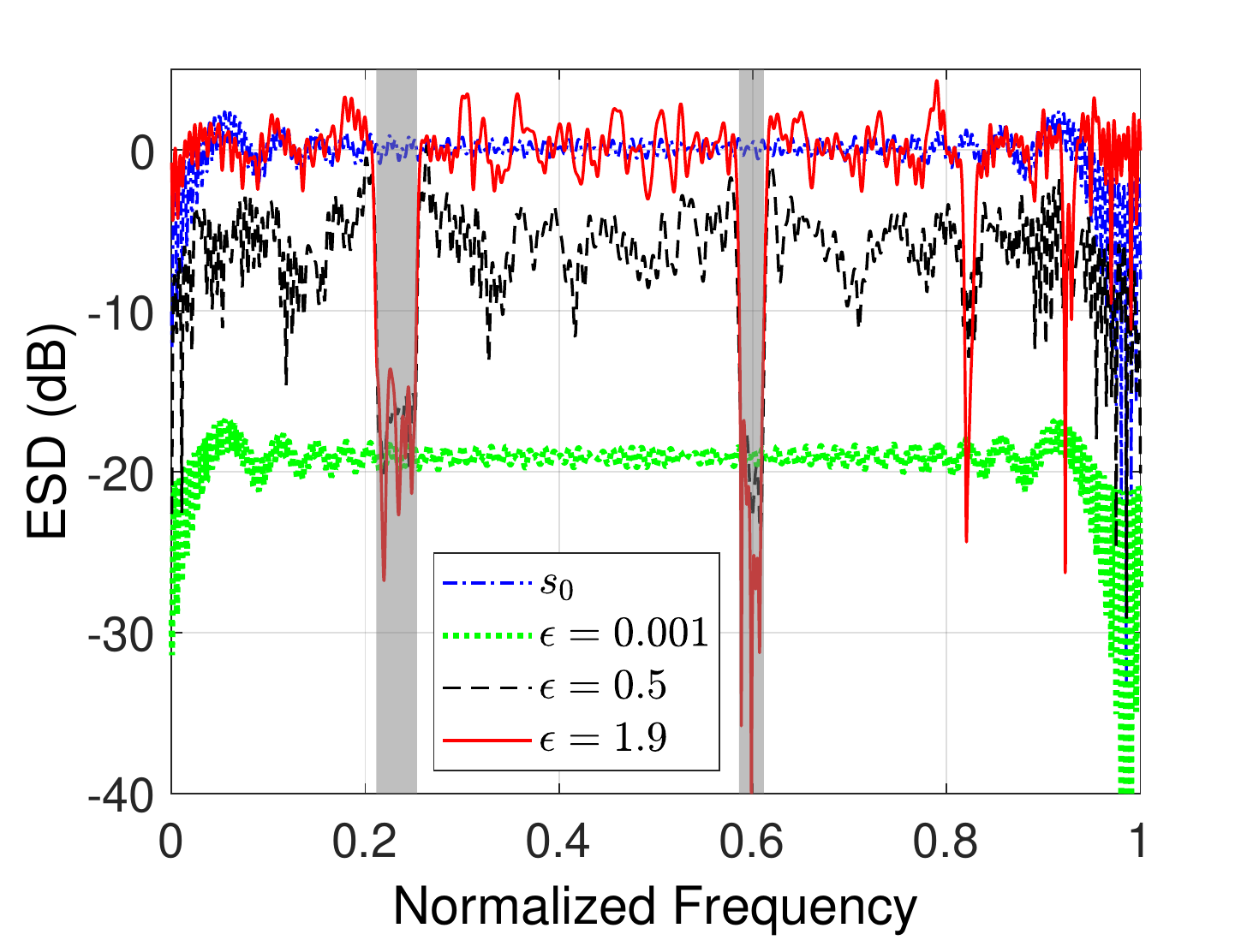}}
	\subfigure[$M=64$]{
		\includegraphics[width=0.95\columnwidth]{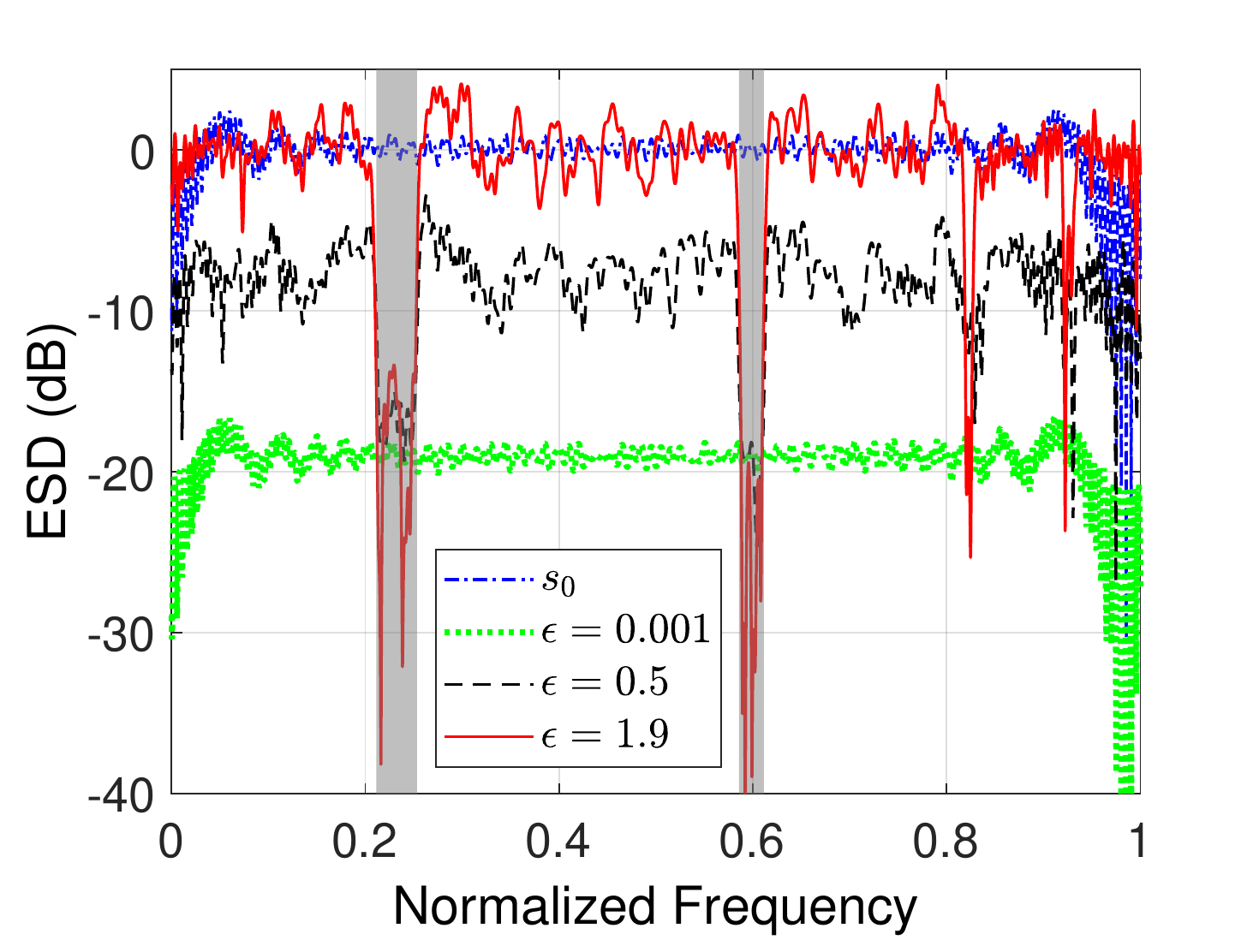}}
	
	\caption{ESD versus normalized frequency of the phase codes designed for $\epsilon=0.001,0.5,1.9$. (a) Continuous phase, (b) $M=64$.}\label{fig:esd}
\end{figure}

To shed light on the capabilities of the devised transceiver to mitigate signal-dependent disturbance, the PSL and ISL of
the Cross-Correlation Function (CCF) normalized to $|\bw^{\star\dag}\bs^{\star}|$ between
the transmit waveform and the receive filter versus the iteration index, are reported in Table II for the continuous phase and $M=64$, assuming $\epsilon= 1$ and HIVAM as the initialization method. The results show that lower and lower PSL and ISL values are achieved as the iteration step $n$ grows up for both continuous and discrete phase codes.  Otherwise stated, the devised strategy is able to iteratively improve the rejection of the signal-dependent interference.

\begin{table}[h] \label{tablepsl}
	\caption{PSL and ISL of the normalized CCFs between the transmit waveform and the receive filter, at different iterations, assuming $\epsilon=1$.}
		\centering
	\begin{tabular}{cccccclllll}
\multicolumn{6}{c}{\begin{tabular}[c]{@{}c@{}}(a) Continous phase\\ (the number of iterations  up to convergence is 15)\end{tabular}}                                                                                                                                             &                      &                      &                      &                      &                      \\ \cline{1-6}
		\multicolumn{1}{|c|}{$n$}     & \multicolumn{1}{c|}{0}      & \multicolumn{1}{c|}{2}     & \multicolumn{1}{c|}{4}     & \multicolumn{1}{c|}{8}     & \multicolumn{1}{c|}{15}     &                      &                      &                      &                      &                      \\ \cline{1-6}
		\multicolumn{1}{|c|}{PSL(dB)} & \multicolumn{1}{c|}{-18.25} & \multicolumn{1}{c|}{-19.63} & \multicolumn{1}{c|}{-20.39} & \multicolumn{1}{c|}{-21.19} & \multicolumn{1}{c|}{-21.22} &                      &                      &                      &                      &                      \\ \cline{1-6}
		\multicolumn{1}{|c|}{ISL(dB)} & \multicolumn{1}{c|}{-6.95}  & \multicolumn{1}{c|}{-7.78}  & \multicolumn{1}{c|}{-8.12}  & \multicolumn{1}{c|}{-8.47}  & \multicolumn{1}{c|}{-8.52}  &                      &                      &                      &                      &                      \\ \cline{1-6}
\multicolumn{6}{c}{\begin{tabular}[c]{@{}c@{}}(b) $M=64$\\ (the number of iterations  up to convergence is 20)  \end{tabular}}                                                                                                                                                    & \multicolumn{1}{c}{} & \multicolumn{1}{c}{} & \multicolumn{1}{c}{} & \multicolumn{1}{c}{} & \multicolumn{1}{c}{} \\ \cline{1-6}
		\multicolumn{1}{|c|}{$n$}     & \multicolumn{1}{c|}{0}      & \multicolumn{1}{c|}{2}      & \multicolumn{1}{c|}{5}     & \multicolumn{1}{c|}{10}     & \multicolumn{1}{c|}{20}     &                      &                      &                      &                      &                      \\ \cline{1-6}
		\multicolumn{1}{|c|}{PSL(dB)} & \multicolumn{1}{c|}{-18.39} & \multicolumn{1}{c|}{-20.24} & \multicolumn{1}{c|}{-20.96} & \multicolumn{1}{c|}{-21.08} & \multicolumn{1}{c|}{-21.34} &                      &                      &                      &                      &                      \\ \cline{1-6}
		\multicolumn{1}{|c|}{ISL(dB)} & \multicolumn{1}{c|}{-6.49}  & \multicolumn{1}{c|}{-7.89}  & \multicolumn{1}{c|}{-8.18}  & \multicolumn{1}{c|}{-8.32}  & \multicolumn{1}{c|}{-8.39}  &                      &                      &                      &                      &                      \\ \cline{1-6}
	\end{tabular}
\end{table}

\subsection{Comparison with Available Algorithms}\label{com}

If spectral coexistence requirements are not considered, the resulting problem can be solved by Dinkelbach-Type Algorithm (DTA) \cite{b1} and Majorized Iterative Algorithm
with the Constant Modulus and Similarity Constraint (MIA-CMSC) \cite{b2}, which can be further distinguished into MIA-CMSC1 and MIA-CMSC2 depending on different majorization methods\footnote{The parameters specified in \cite{b1} and \cite{b2} are used to implement these algorithms; specifically, $\varsigma=10^{-3}$ and $\kappa=10^{-4}$ for DTA and $\epsilon_{obj}=10^{-4}$ for MIA-CMSC1 and MIA-CMSC2.}.
Figs. \ref{fig:qp}(a)-(b) depict the $\overline{\mbox{SINR}}$  achieved by {\bf Algorithm \ref{alg_1}}, DTA, MIA-CMSC1, and MIA-CMSC2 assuming at the design stage a continuous phase and $M=32$, respectively\footnote{For all the algorithms at each $\epsilon$ both the reference sequence $\bs_0$ and the code optimized at the previous $\epsilon$ are adopted as the initializations. The one providing the highest $\overline{\mbox{SINR}}$ between the two synthesized sequences is picked up. Although DTA and MIA-CMSC are not provided in \cite{b1} and \cite{b2} with reference to discrete phase codes, their extension to encompass also this design constraint is straightforward.}.  
 Looking over the figures unveils that the proposed CD framework outperforms the counterparts. The superior performance with respect to DTA can be attributed to the slight phase-search sub-optimality of the DTA, which is intrinsic in the Dinkelbach iterative method. Instead, the performance loss incurred by the MIA-CMSC methods reasonably results from the approximation of the objective function performed in the procedure. Finally, in the design of discrete phase codes  MIA-CMSC procedures experience a higher performance gap with respect to \textbf{Algorithm \ref{alg_1}} than that observed for the continuous phase instance.

 To assess the convergence property and computational complexity of the different methods,
 Figs. \ref{fig:convergence}(a)-(b) depict the $\overline{\mbox{SINR}}$  versus CPU time with $\epsilon=0.8$ for the continuous phase and $M=32$, respectively.  The reference code is used for this case study to initialize all the algorithms. As expected, the objective function of all the methods monotonically increases. Inspection of the curves shows that \textbf{Algorithm \ref{alg_1}} is substantially capable of obtaining larger $\overline{\mbox{SINR}}$ values than the counterparts for any given time budgets, which provides a practical proof of the CD framework efficiency. Specifically, \textbf{Algorithm \ref{alg_1}} outperforms the counterparts, but for the continuous phase code synthesis and up to 0.38s where MIA-CMSC1 is better.
\begin{figure}
	\centering
	\subfigure[Continuous phase]{
		\includegraphics[width=0.9\columnwidth]{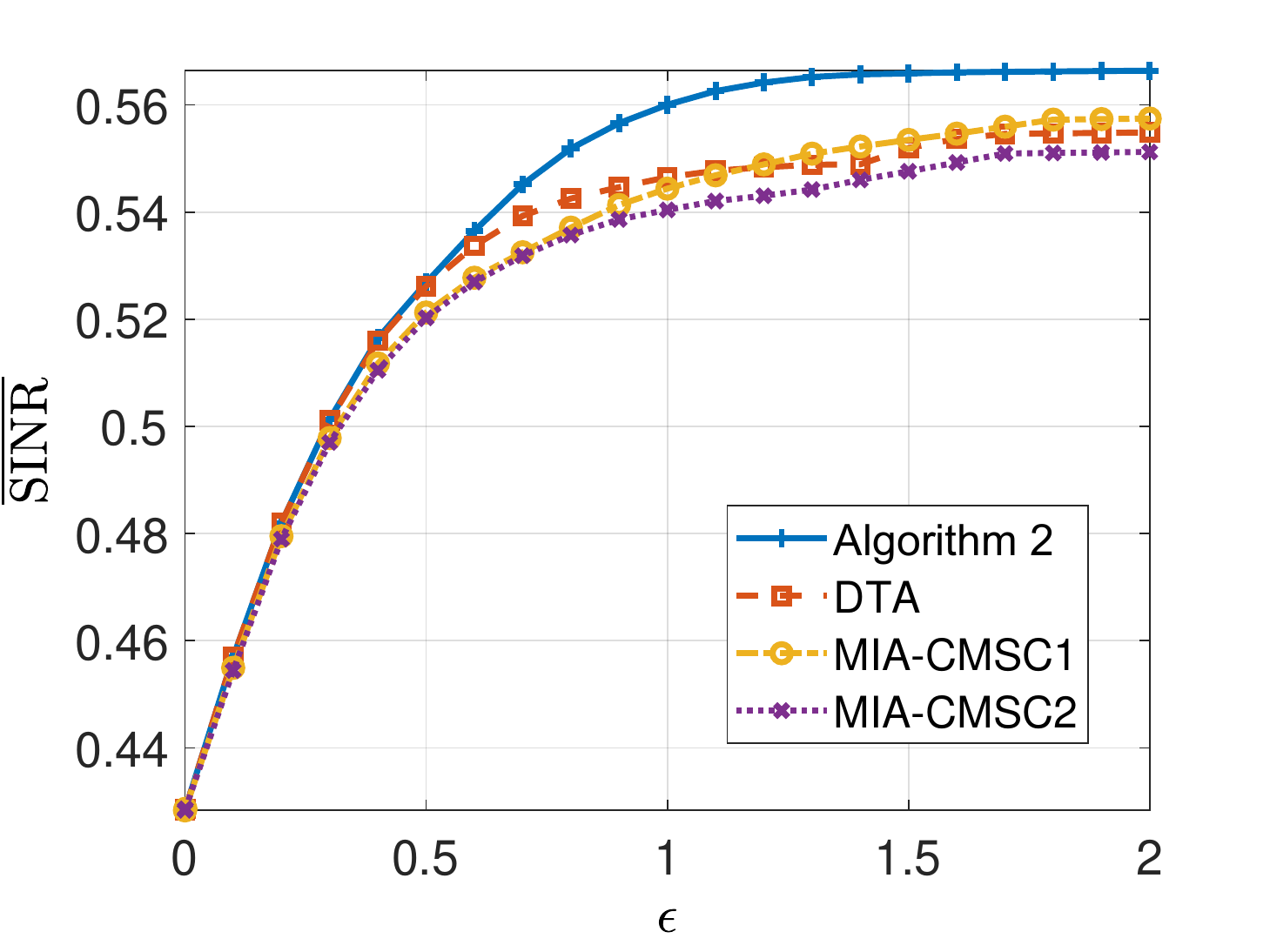}}	\subfigure[$M=32$]{
		\includegraphics[width=0.9\columnwidth]{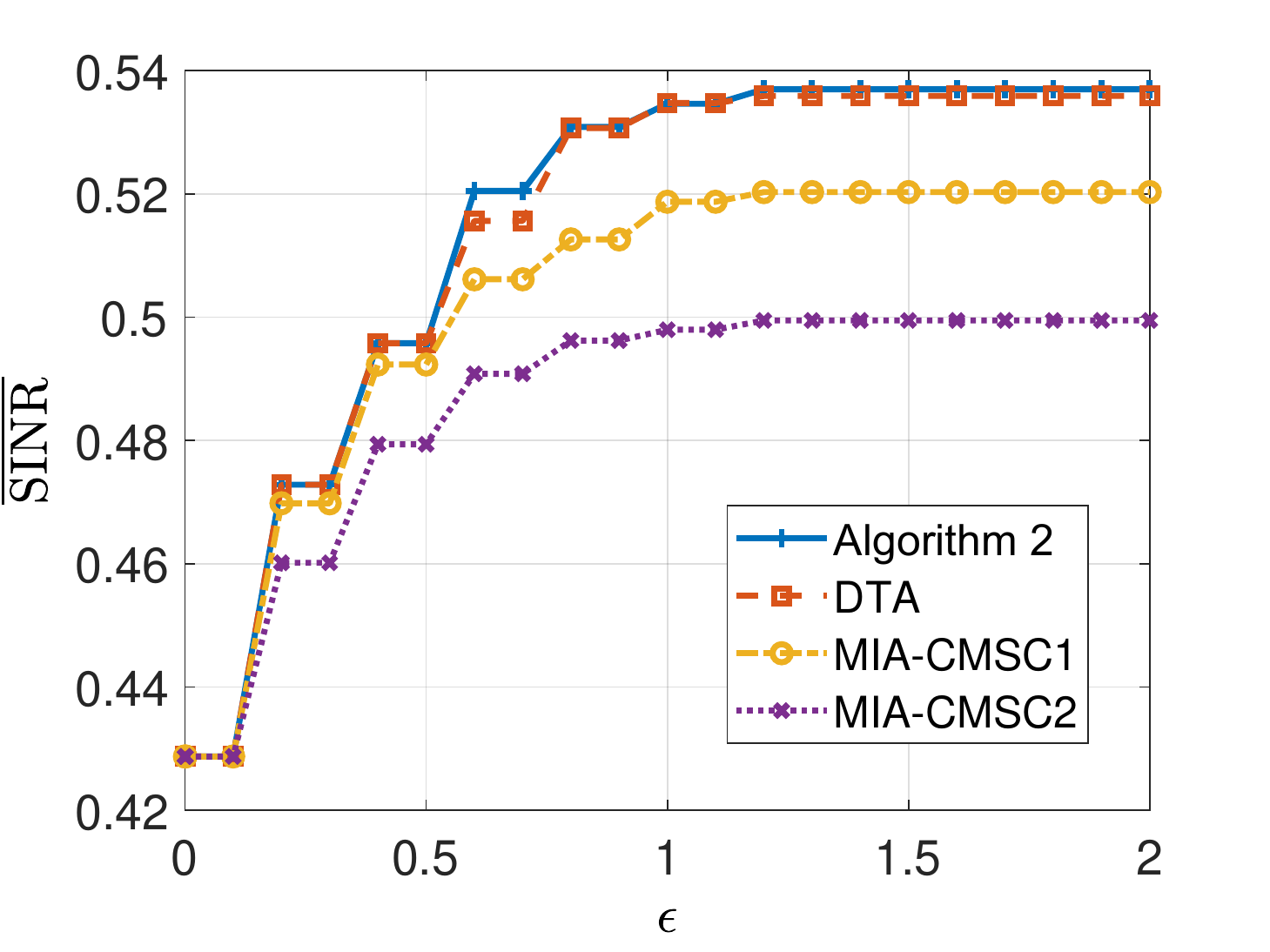}}
	\caption{Achieved $\overline{\mbox{SINR}}$ versus $\epsilon$ without spectral constraints. (a) Continuous phase, (b) $M=32$.}\label{fig:qp}
\end{figure}

 \begin{figure}
	\centering
	\subfigure[Continuous phase]{
		\includegraphics[width=0.9\columnwidth]{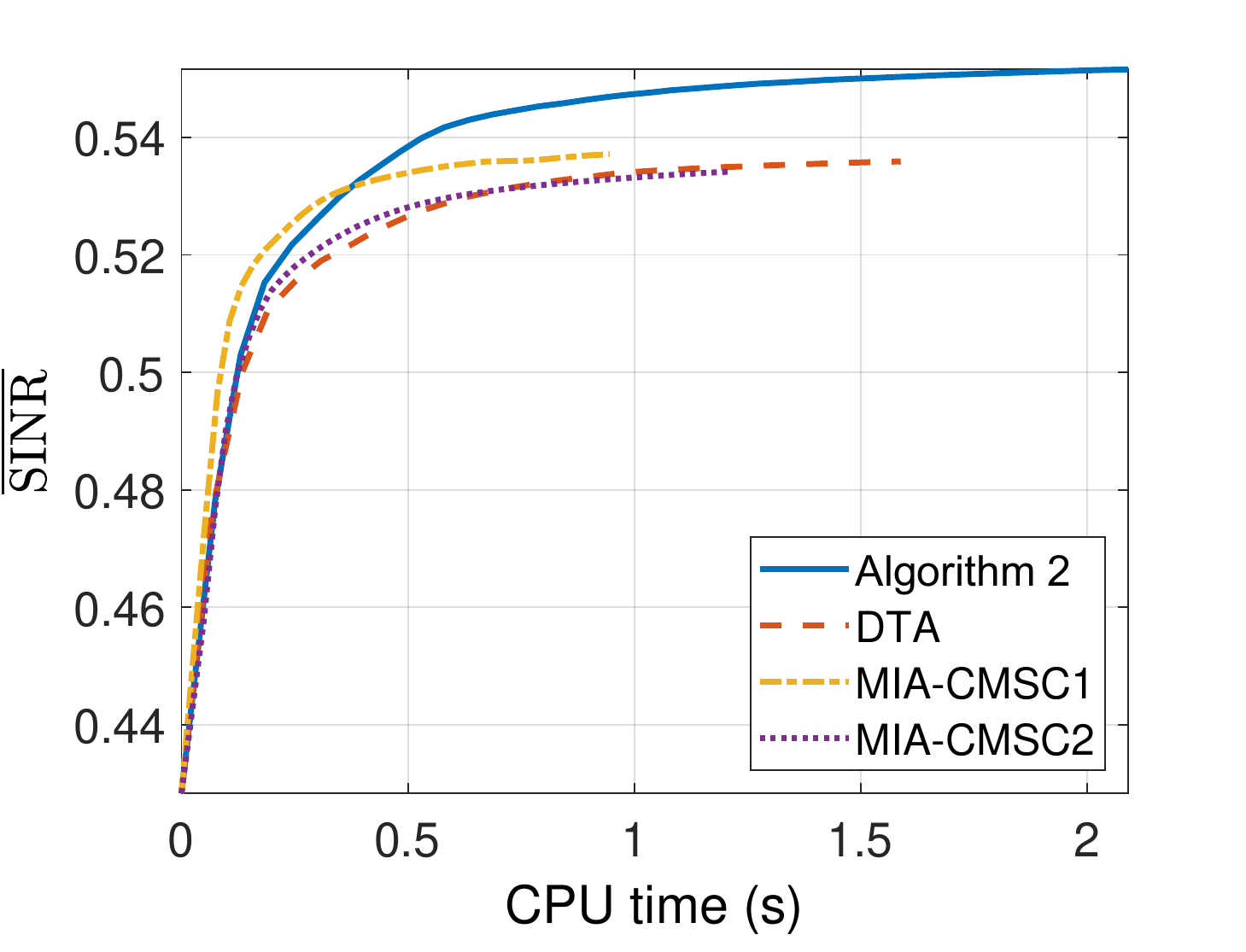}}
	\subfigure[$M=32$]{
		\includegraphics[width=0.9\columnwidth]{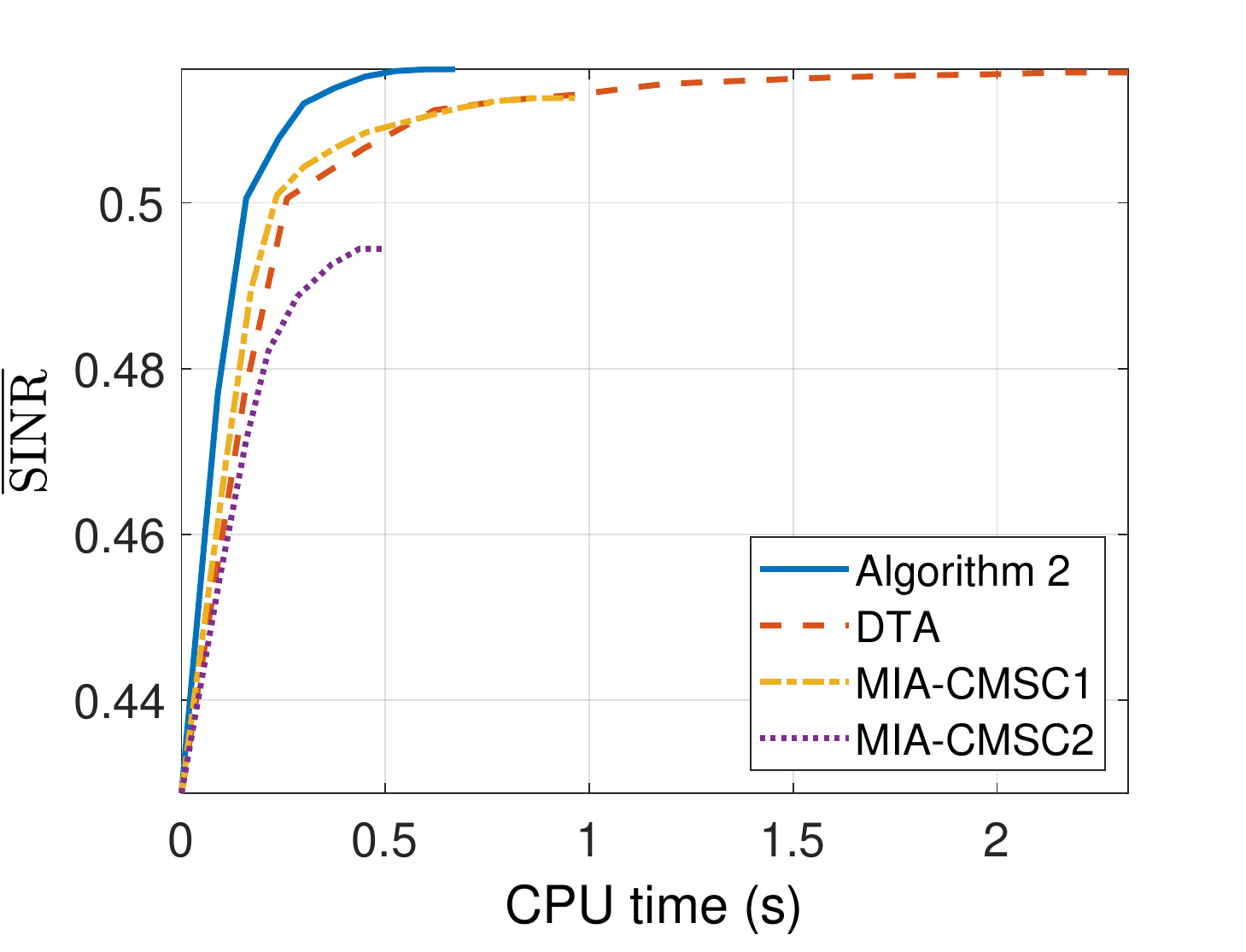}}
	\caption{ $\overline{\mbox{SINR}}$ versus CPU time with $\epsilon=0.8$. (a) Continuous phase, (b) $M=32$.}\label{fig:convergence}
\end{figure}

%
%Dropping the similarity constraint and instituting the constant modulus constraint with $|s_i|<=\frac{1}{\sqrt{N}},i\in {\cal N}$, the resulting problem for continuous phase case can be solved by Majorized Iterative Algorithm
%with the Spectrum
%Compatibility Constraint for Local design (MIA-SCCL)\cite{b2},  and a  Semidefinite Programming
% based sequence Optimization (SDPD) method in \cite{aa}.
%In Fig. \ref{fig:mup},  compared with the proposed CD algorithm,  MIA-SCCL and SDPD cost more time to attain a larger $\overline{\mbox{SINR}}$ as CVX toolbox \cite{cvx} is adopted.  Fig. \ref{fig:cote} shows the amplitudes of the transmit signals derived by the three algorithms are very close, while a looser amplitude constraint is forced in MIA-SCCL and SDPD compared with CD, which causes the slight fluctuation. 
To proceed further, let us observe that, neglecting the constant envelope as well as the similarity  requirements, ${\cal{P}}_c$ with $K=2$ can be solved via  Majorized Iterative Algorithm with the Spectrum Compatibility Constraint for Local design (MIA-SCCL)\cite{b2}, as well as a variant of the algorithm in \cite{aa}, denoted as SemiDefinite
Programming-based Design (SDPD)\footnote{The  parameters specified in \cite{b2} and \cite{aa} are used for implementing these algorithms; specifically, $\epsilon_{slc}=10^{-8}$ and $\epsilon_{obj}=10^{-4}$ for  MIA-SCCL and  $\zeta=10^{-4}$ for SDPD.}.  The  $\overline{\mbox{SINR}}$, the computational time, and the PAR \cite{8454704} of the sequences synthesized via SDPD, MIA-SCCL and {\bf Algorithm \ref{alg_1}}  are summarized in the Table \ref{table22}. The sequence devised via HIVAC is adopted as the initialization for all the algorithms, which costs 1.8724s for computation. The proposed CD algorithm is capable of devising the radar transceiver with a shorter running time than the counterparts. However, as expected, it experiences a $\overline{\mbox{SINR}}$ loss with respect to SDPD and MIA-SCCL. Indeed, SDPD and MIA-SCCL may capitalize on code amplitude variation to boost the radar performance, at the price of larger PAR values, which in turn demand for more sophisticated amplifiers.

% \begin{figure}
%	\centering
%	\includegraphics[width=0.95\columnwidth]{figure/Amplitude}	
%	\caption{ Amplitude of the designed signals.}\label{fig:cote}
%\end{figure}
\begin{table}
	\centering
	\caption{$\overline{\mbox{SINR}}$, corresponding computation time and PAR for different algorithms with multi-spectral constraints.}
	\begin{tabular}{|c|c|c|c|}
		\hline
		Algorithm & $\overline{\mbox{SINR}}$  & Time (s)& PAR    \\ \hline
		SDPD      & 0.6088                   & 855.4809& 2.4520  \\ \hline
		MIA-SCCL  & 0.6004                    & 6.6539 & 3.3589 \\ \hline
		CD        & 0.5120                 & 1.7256 & 1       \\ \hline
	\end{tabular}
\label{table22}
\end{table}

\section{Conclusion}\label{sec_conclusion}
The synthesis of radar transceivers  in signal-dependent interference and spectrally
contested-congested environments has been addressed in this paper. Specifically, assuming constant modulus signals  with either continuous or finite alphabet phases, 
 the design has aimed at maximizing the SINR at the output of the receive filter while ensuring cohabitation with surrounding RF systems via multiple spectral constraints. 
Furthermore, a similarity constraint has been forced on the probing signal
to bestow attractive waveform characteristics. Hence,
 resorting to the CD framework,  iterative procedures (with ensured convergence properties) have been conceived to synthesize optimized radar waveforms and receive filters. 
 Each step requires the solution of a possible non-convex optimization problem whose global optimal point is obtained in closed-form, regardless of the phase codes cardinality.
 Remarkably, the overall computational complexity of the proposed algorithms is
	polynomial with respect to both the code length and the number of licensed emitters. 

At the analysis stage the performance of the synthesized transceiver pair has been assessed in terms of different metrics, i.e.,  SINR as well as detection probability, spectral shape, and CCF features. The results have clearly shown that the proposed framework is capable of mitigating diverse interfering scenarios, while ensuring cohabitation with overlaid licensed systems. Besides, some comparisons with counterparts available in the open literature have been provided, showing interesting performance gains, in terms of achieved SINR and/or computational time.

Future research tracks might concern the extension of this framework to  MIMO radar systems 
where spatial diversity may induce other degrees of freedom which can be optimized to improve the
performance of the overall method.

\section*{\center{ \textnormal{ACKNOWLEDGMENT}}}
This research activity has been conducted during the visit
of Jing Yang at the University of Napoli ``Federico II" DIETI under the local supervision of Prof. A. Aubry and Prof. A. De Maio. The Authors
would like to thank Dr. Linlong Wu, Prof. Prabhu Babu, and Prof. Daniel P. Palomar for
providing the MATLAB code for MIA.

\appendix
\subsection{Derivation of Problem \eqref{phase}}\label{A}
Before proceeding further, let us introduce the following lemma.
\begin{lemma}\label{le1}
	For any known matrix $\bR\in {\mathbb{H}}^{N}$,  $g\left( \ba\right) =\ba^\dag\bR\ba$ with $\ba\in {\mathbb{C}}^{N}$ can be recast as a function of a specific entry $a_h,h\in\cal{N}$, i.e.,
	\begin{equation}
	\begin{split}
	g\left({a}_h;{\ba}_{-h}\right)=&(\hat {\ba}_{-h}+{a}_h\bm e_{h})^\dag\bR(\hat {\ba}_{-h}+{a}_h\bm e_{h})\\
	=&\bm e_{h}^\dag{\bR}\bm e_{h}|a_h^{(n)}|^2+2\Re\left\{
	\hat {\ba}^\dag_{-h}\bR\bm e_{h}{a}_h \right \}\\
	&+	\hat {\ba}^\dag_{-h}\bR\hat {\ba}_{-h},
	\end{split}
	\end{equation}
	where  $\hat {\ba}_{-h}=\ba-{a}_h\bm e_{h}\in {\mathbb{C}}^{N}$, and $\ba_{-h}=[a_1,\ldots,a_{h-1},a_{h+1},\ldots,a_N]^T\in {\mathbb{C}}^{N-1}$.
\end{lemma}

According to \textbf{Lemma} \ref{le1} and exploiting $|{\bar s}_i|^2=|e^{jy_i}|^2=1,i\in\cal N$ as well as $y_{N+1}\ge0$, $\chi\left(y_h;{\by}_{-h}^{(n)}\right)$ w.r.t. $y_h$ can be expressed as
\begin{flalign}
	\chi\left(y_h;{\by}_{-h}^{(n)}\right)=\frac{\Re\{a_{h}^{(n)}e^{jy_h}\}+b_h^{(n)} }{\Re\{c_h^{(n)}e^{jy_h}\}+d_{h}^{(n)} },h\in{\cal{N}},
\end{flalign}
where 
\begin{itemize}
	\item
	$a_{h}^{(n)}=2{\hat \bs}^{(n)\dag}_{-h}\bM_1^{\left( \by_{N+2}^{(n-1)}\right)} \bm e_h$,
	\item $b_h^{(n)}={\hat \bs}^{(n)\dag}_{-h}\bM_1^{\left( \by_{N+2}^{(n-1)}\right)}{\hat \bs}^{(n)}_{-h}+\bm e_{h}^\dag\bM_1^{\left( \by_{N+2}^{(n-1)}\right)}\bm e_{h}$,
	\item
	$c_h^{(n)}=2{\hat \bs}^{(n)\dag}_{-h}\bM_2^{\left( \by_{N+2}^{(n-1)}\right)}\bm e_h$,
	\item
	$d_{h}^{(n)}={\hat \bs}^{(n)\dag}_{-h}\bM_2^{\left( \by_{N+2}^{(n-1)}\right)}{\hat \bs}^{(n)}_{-h}+\bm e_{h}^\dag\bM_2^{\left( \by_{N+2}^{(n-1)}\right)}\bm e_{h}+{\vartheta^{\left( \by_{N+2}^{(n-1)}\right)}}/{{y_{N+1}^{(n-1)}}}$, 
\end{itemize}
with ${\hat\bs}^{(n)}_{-h}=\left[e^{jy_1^{(n)}},\ldots,e^{jy_{h-1}^{(n)}},0,e^{jy_{h+1}^{(n-1)}},\ldots, e^{jy_{N}^{(n-1)}}\right] ^T$.
%${\hat \bs}^{(n)}_{-h}={\bar \bs}_{\bm \varphi}_h^{(n)}-{x}_h\bm e_{h}$.

Similarly, the spectral constraints can be transformed as
\begin{equation}\label{spc}
\begin{split}
{y_{N+1}^{(n-1)}}\left(\bar{z}_{k,h}^{(n)}+ \Re\left\{z_{k,h}^{(n)}\,e^{jy_h}\right\}\right)\le E_I^k,k=1,\ldots,K,
\end{split}
\end{equation}
where $\bar{z}_{k,h}^{(n)}={\hat \bs}^{(n)\dag}_{-h}\overline{{\bR}}_{I}^k{\hat \bs}^{(n)}_{-h}+\bm e_{h}^\dag\overline{{\bR}}_{I}^k\bm e_{h}$, $z_{k,h}^{(n)}=2{\hat \bs}^{(n)\dag}_{-h}\overline{{\bR}}_{I}^k\bm e_{h}$. As a result, being $y_{N+1}\ge0$, the inequalities in \eqref{spc} are  tantamount to $\Re\{z_{k,h}^{(n)}e^{jy_h}\}\le \tilde{c}_{k,h}^{(n)}, k=1,\ldots,K$, where $\tilde{c}_{k,h}^{(n)}=\frac{E_I^k}{ {y_{N+1}^{(n-1)}}}-{\bar{z}_{k,h}^{(n)}} \in \mathbb{R}$.

\subsection{Monotonicity Study of $\chi\left(y_h;{\by}_{-h}^{(n)}\right)$ and Evaluation of $\bar{\mathcal{F}}_p$}\label{proproof}
Before proceeding further, let us observe that both the characterization  of the objective function monotonicities and the feasible set derivation can be performed by means of a change of variable, that defines a one-to-one monotonically increasing mapping. To this end, let us consider $t=\tan(y_h/2)$. In the transformed domain, the objective function of Problem \eqref{phase} can be rewritten as
\begin{equation}\label{rt}
R_h^{(n)}(t)=\frac{a_{1,h}^{(n)}t^2+b_{1,h}^{(n)}t+c_{1,h}^{(n)}}{a_{2,h}^{(n)}t^2+b_{2,h}^{(n)}t+c_{2,h}^{(n)}},
\end{equation}
where
\begin{equation}\label{abc}
\begin{split}
a_{1,h}^{(n)}=b_h^{(n)}-\Re\{a_{h}^{(n)}\},&a_{2,h}^{(n)}=d_{h}^{(n)}-\Re\{c_h^{(n)}\},\\
b_{1,h}^{(n)}=-2\Im\{a_{h}^{(n)}\},&b_{2,h}^{(n)}=-2\Im\{c_h^{(n)}\},\\
c_{1,h}^{(n)}=b_h^{(n)}+\Re\{a_{h}^{(n)}\},&c_{2,h}^{(n)}=d_{h}^{(n)}+\Re\{c_h^{(n)}\}.
\end{split}
\end{equation}

Similarly, $ \Re\{z_{k,h}^{(n)}e^{jy_h}\}\le \tilde{c}_{k,h}^{(n)}$ can be cast as
\begin{equation}\label{sp}
\bar{a}_{k,h}^{(n)}t^2+\bar{b}_{k,h}^{(n)}t+\bar{c}_{k,h}^{(n)}\le0,
\end{equation}
where 
\begin{eqnarray}\label{abck}
\bar{a}_{k,h}^{(n)}&=&-\Re\{z_{k,h}^{(n)}\}- \tilde{c}_{k,h}^{(n)},\nonumber\\
\bar{b}_{k,h}^{(n)}&=&-2\Im\{z_{k,h}^{(n)}\},\\
\bar{c}_{k,h}^{(n)}&=&\Re\{z_{k,h}^{(n)}\}- \tilde{c}_{k,h}^{(n)}.\nonumber
\end{eqnarray}

%Problem \eqref{phase} with either $p=\infty$ or $p=M$ can be equivalently recast as
%	\begin{equation*}\label{refor}
%		{\cal{P}}_{h,n,t}^p\left\{ \begin{array}{ll}
%			\displaystyle{\max_{t}} & R_h^{(n)}(t)=\frac{a_{1,h}^{(n)}t^2+b_{1,h}^{(n)}t+c_{1,h}^{(n)}}{a_{2,h}^{(n)}t^2+b_{2,h}^{(n)}t+c_{2,h}^{(n)}}\\
%			\rm{s.t.}  & {\bar{a}_{k,h}^{(n)}t^2+\bar{b}_{k,h}^{(n)}t+\bar{c}_{k,h}^{(n)}\le0},k=1,\ldots,K\\
%			& t \in {\bar\Psi}_p
%		\end{array} \right.
%	\end{equation*}
%	where  	

Let us first characterize the behavior of the objective function $R_h^{(n)}(t)$.
To this end, note that ${a_{2,h}^{(n)}t^2+b_{2,h}^{(n)}t+c_{2,h}^{(n)}}>0,\forall t$, which implies that either $a_{2,h}^{(n)}>0$ and $\sqrt{b_{2,h}^{(n)2}-4a_{2,h}^{(n)}c_{2,h}^{(n)}}\le 0$, or $a_{2,h}^{(n)}=b_{2,h}^{(n)}=0$ and $c_{2,h}^{(n)}>0$.
As to the first-order derivation of $R_h^{(n)}(t)$, it is given by
\begin{equation}\label{deri}
R_h^{(n)'}(t)=\frac{\hat{d}_h^{(n)}t^2+\hat{e}_h^{(n)}t+\hat{f}_h^{(n)}}{\left({a_{2,h}^{(n)}t^2+b_{2,h}^{(n)}t+c_{2,h}^{(n)}}\right)^2},
\end{equation}
where $\hat{d}_h^{(n)}=a_{1,h}^{(n)}b_{2,h}^{(n)}-a_{2,h}^{(n)}b_{1,h}^{(n)}$, $\hat{e}_h^{(n)}=2(a_{1,h}^{(n)}c_{2,h}^{(n)}-a_{2,h}^{(n)}c_{1,h}^{(n)})$ and $\hat{f}_h^{(n)}=b_{1,h}^{(n)}c_{2,h}^{(n)}-b_{2,h}^{(n)}c_{1,h}^{(n)}$. According to \textbf{Lemma} \ref{lemma} (reported below), if $\hat{d}_h^{(n)}\neq0$, $R_h^{(n)}(t)$ admits two stationary points. In particular, $R_h^{(n)}(t)$ exhibits the following behavior:
\begin{itemize}
	\item if $\hat{d}_h^{(n)}=0$, it follows that
	\begin{itemize}
		\item if $\hat{e}_h^{(n)}=0$, it follows that $\hat{f}_h^{(n)}=0$, implying that $R_h^{(n)}(t)$ is a constant function;
		\item if $\hat{e}_h^{(n)}>0$ (see Fig. \ref{fig} (a) for a notional example), $R_h^{(n)}(t)$ is strictly decreasing over $t<-\hat{f}_h^{(n)}/\hat{e}_h^{(n)}$, and strictly increasing over $t>-\hat{f}_h^{(n)}/\hat{e}_h^{(n)}$.
		Thus, the minimum point is $t_{s,h}^{(n)}=-\hat{f}_h^{(n)}/\hat{e}_h^{(n)}$, and the finite supremum is achieved at $\pm\infty$.
		\item if $\hat{e}_h^{(n)}<0$ (see Fig. \ref{fig} (b) for a notional example), $R_h^{(n)}(t)$ is strictly increasing over $t<-\hat{f}_h^{(n)}/\hat{e}_h^{(n)}$, and strictly decreasing over $t>-\hat{f}_h^{(n)}/\hat{e}_h^{(n)}$. Thus, the maximum point is $t_{g,h}^{(n)}=-\hat{f}_h^{(n)}/\hat{e}_h^{(n)}$, and the finite infimum is achieved at $\pm\infty$;
	\end{itemize}
	\item if $\hat{d}_h^{(n)}\neq0$, the roots of \eqref{deri} are ruled by the discriminant $\Delta_h^{(n)}=\hat{e}_h^{(n)2}-4\hat{d}_h^{(n)}\hat{f}_h^{(n)}$, which are given by $t_{g,h}^{(n)}=\frac{-\hat{e}_h^{(n)}-\sqrt{\Delta_h^{(n)}}}{2\hat{d}_h^{(n)}}$,  $t_{s,h}^{(n)}=\frac{-\hat{e}_h^{(n)}+\sqrt{\Delta_h^{(n)}}}{2\hat{d}_h^{(n)}}$;
	\begin{itemize}
		\item
		if $\hat{d}_h^{(n)}>0$ (see Fig. \ref{fig} (c) for a notional example), $R_h^{(n)}(t)$ is strictly increasing over $t<t_{g,h}^{(n)}$, strictly decreasing over $t_{g,h}^{(n)}<t<t_{s,h}^{(n)}$, and strictly increasing over $t>t_{s,h}^{(n)}$;
		\item
		if $\hat{d}_h^{(n)}<0$ (see Fig. \ref{fig} (d) for a notional example), $R_h^{(n)}(t)$ is strictly decreasing over $t<t_{s,h}^{(n)}$, strictly increasing over  $t_{s,h}^{(n)}<t<t_{g,h}^{(n)}$, and strictly decreasing over $t>t_{g,h}^{(n)}$.
	\end{itemize}
	
\end{itemize}
As already pointed out, since $y_h\in]-\pi,\pi[$ is mapped to $t\in]-\infty,\infty[$ by a strictly increasing function, the monotonicities  of $\chi\left(y_h;{\by}_{-h}^{(n)}\right),y_h\in]-\pi,\pi[$ can be directly derived starting from those of $R_h^{(n)}(t)$. For instance, 
supposing $\hat{d}_h^{(n)}>0$, $\chi\left(y_h;{\by}_{-h}^{(n)}\right)$ is strictly increasing over $-\pi<y_h<\phi_{g,h}^{(n)}$ with $\phi_{g,h}^{(n)}=2\arctan (t_{g,h}^{(n)})$, then decreasing over $\phi_{g,h}^{(n)}<y_h<\phi_{s,h}^{(n)}$ with $\phi_{s,h}^{(n)}=2\arctan (t_{s,h}^{(n)})$, and increasing over $\phi_{s,h}^{(n)}<y_h<\pi$.
As a result, the maximum and minimum points $\phi_{g,h}^{(n)}$, $\phi_{s,h}^{(n)}$ of the objective function $\chi\left(y_h;{\by}_{-h}^{(n)}\right)$ are\footnote{Note that if $\hat d_h^{(n)} = 0$ and $\hat e_h^{(n)} = 0$, $\chi\left(y_h;{\by}_{-h}^{(n)}\right)$ is a constant function, thus the maximum and minimum points are assigned as $\phi_{g,h}^{(n)}=-\pi$ and $\phi_{s,h}^{(n)}=0$.}
\begin{equation*}
	\phi_{g,h}^{(n)}=\left\{ \begin{array}{ll}
		2\arctan \left( { - {\hat f_h^{(n)}}/{\hat e_h^{(n)}}} \right),&\text{If }\hat d_h^{(n)} = 0,\hat e_h^{(n)} < 0,\\
		- \pi,&\text{If }\hat d_h^{(n)} = 0,\hat e_h^{(n)} > 0,\\
		2\arctan \left( \frac{-\hat{e}_h^{(n)}-\sqrt{\Delta_h^{(n)}}}{2\hat{d}_h^{(n)}}\right),&\text{If }\hat d_h^{(n)} \neq 0,
	\end{array} \right.
\end{equation*}
and
\begin{equation*}
	\phi_{s,h}^{(n)}=\left\{ \begin{array}{ll}
		- \pi ,&\text{If }\hat d_h^{(n)} = 0,\hat e_h^{(n)} < 0,\\
		2\arctan \left( { - \hat f_h^{(n)}/\hat e_h^{(n)}} \right),&\text{If }\hat d_h^{(n)} = 0,\hat e_h^{(n)} > 0,\\
		2\arctan \left( \frac{-\hat{e}_h^{(n)}+\sqrt{\Delta_h^{(n)}}}{2\hat{d}_h^{(n)}}\right),&\text{If }\hat d_h^{(n)} \neq 0.
	\end{array} \right.
\end{equation*}

\begin{figure}
	\centering
	\subfigure[ $\hat{e}_h^{(n)}>0$ and $\hat{d}_h^{(n)}=0$,]{
		\includegraphics[width=0.475\columnwidth]{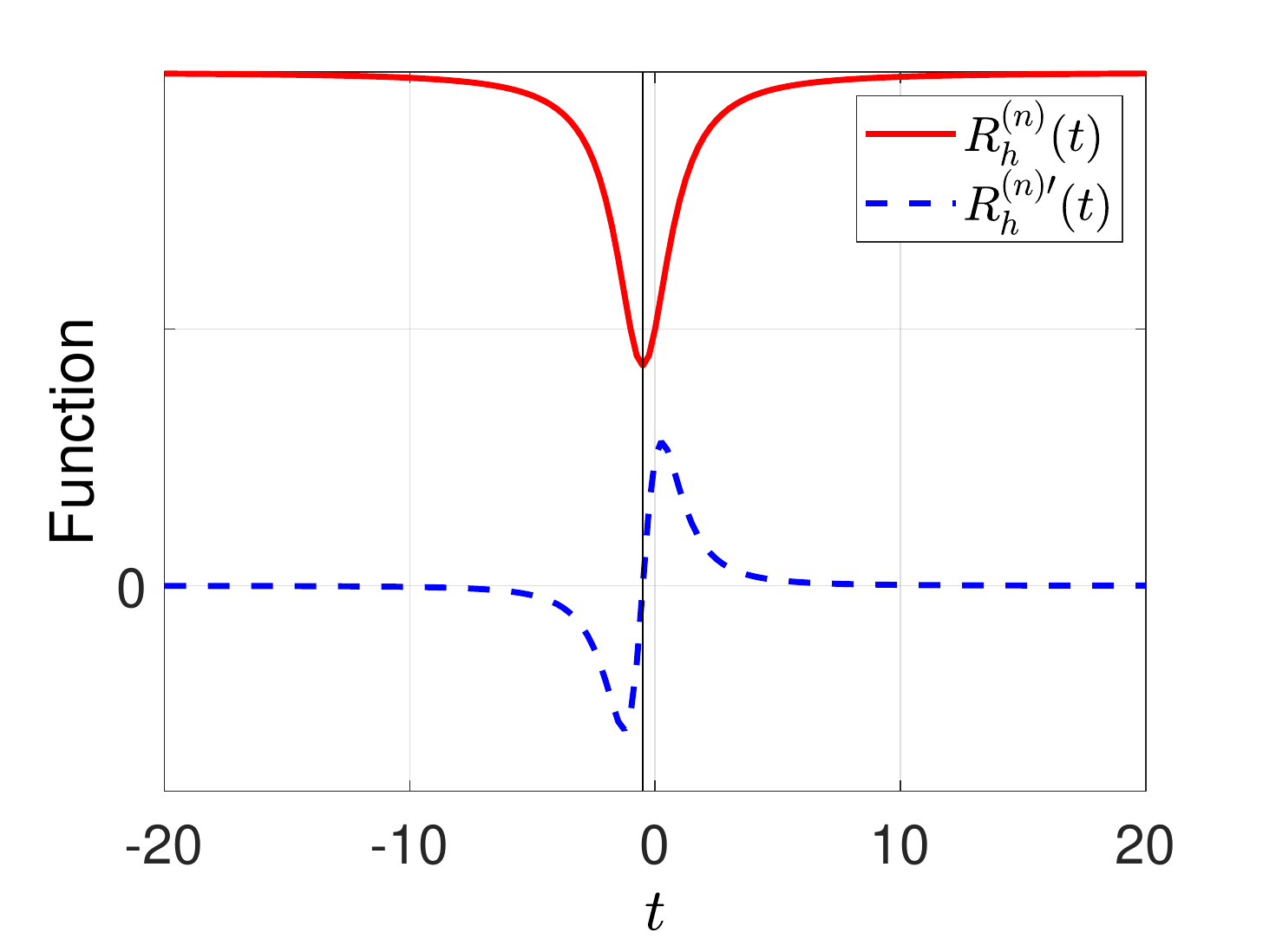}}
	\subfigure[ $\hat{e}_h^{(n)}<0$ and $\hat{d}_h^{(n)}=0$,]{	
		\includegraphics[width=0.475\columnwidth]{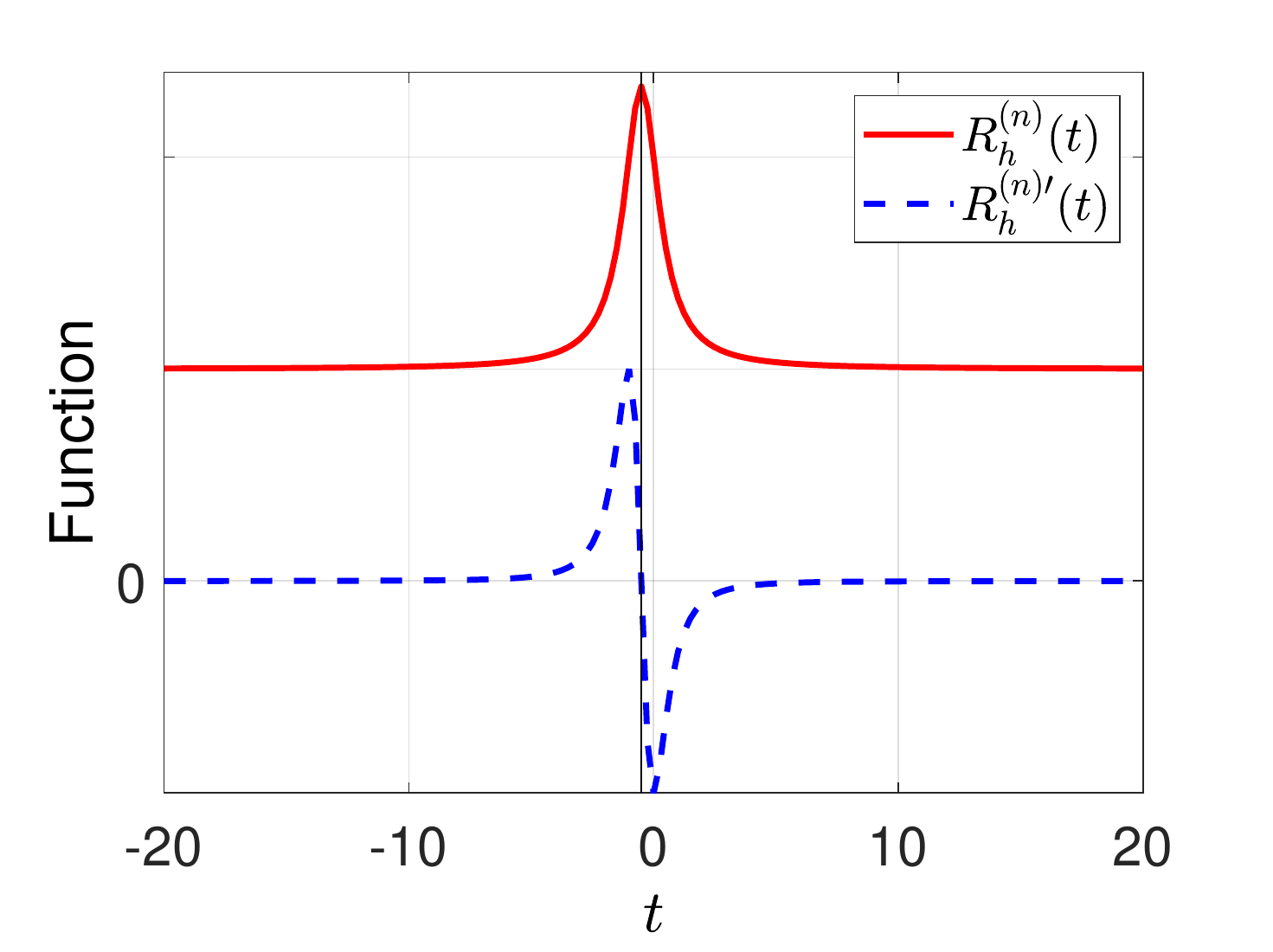}}
	\subfigure[$\hat{d}_h^{(n)}>0$]{
		\includegraphics[width=0.475\columnwidth]{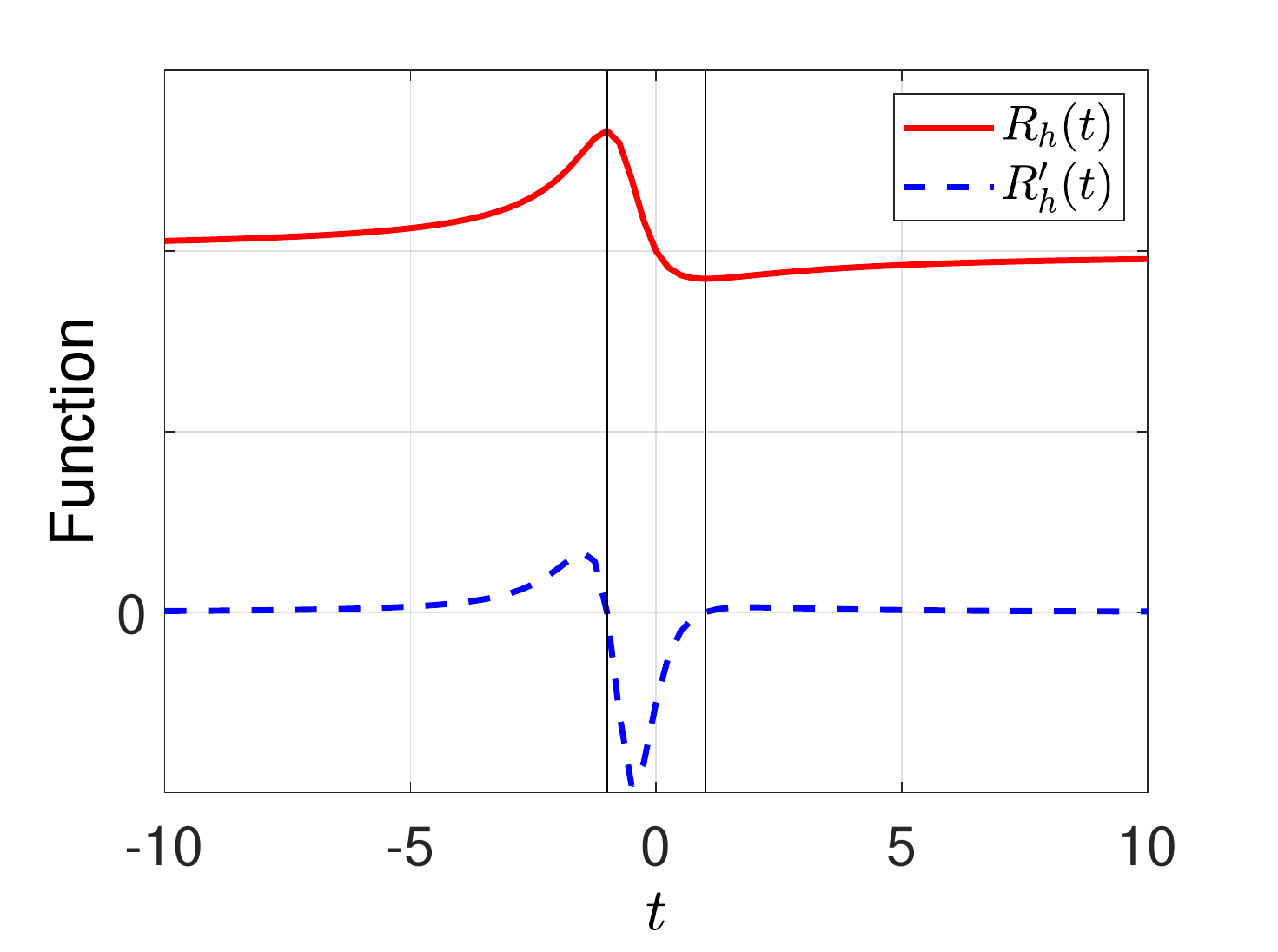}}
	\subfigure[$\hat{d}_h^{(n)}<0$]{
		\includegraphics[width=0.475\columnwidth]{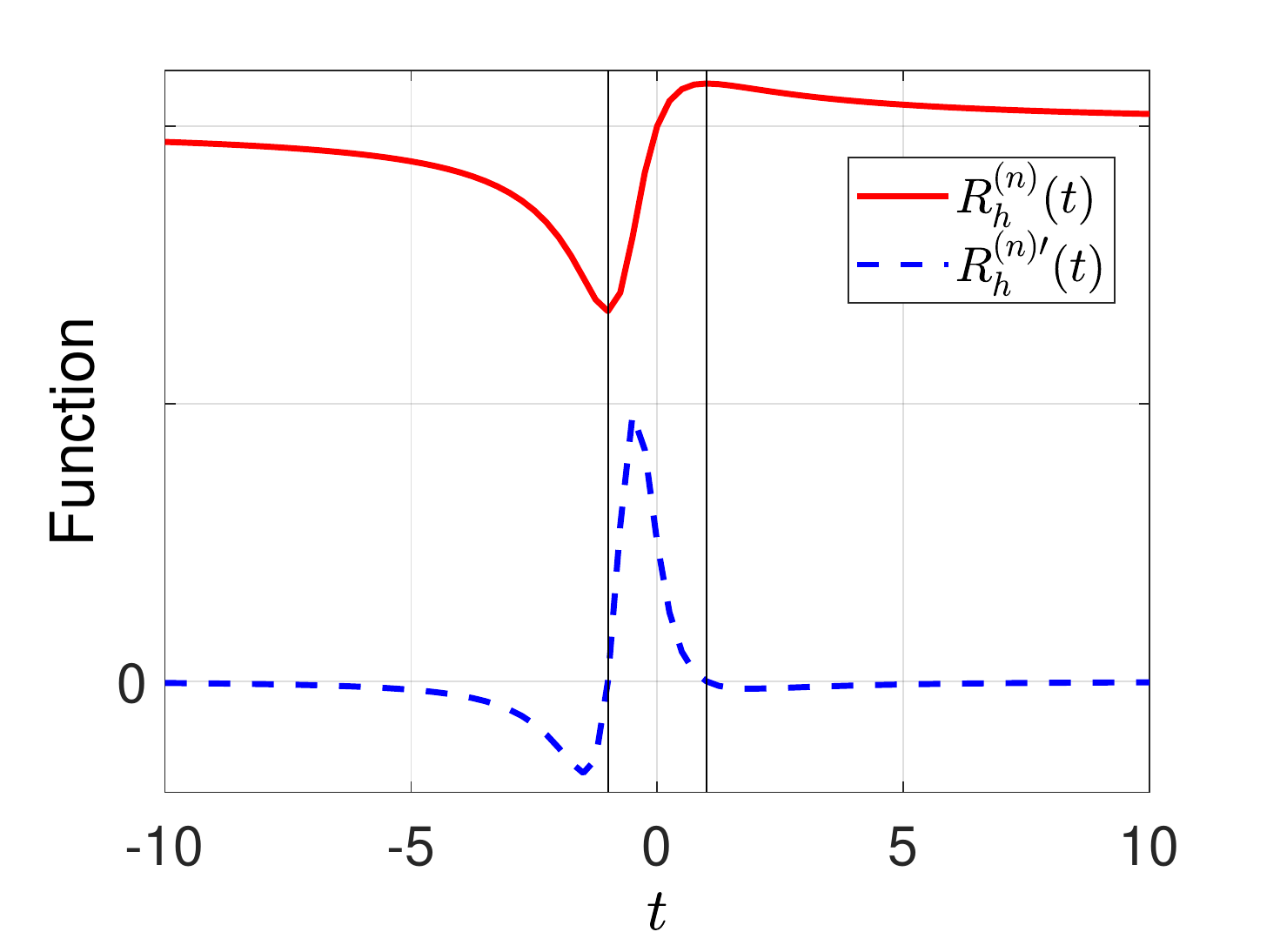}}
	\caption{ Red curve: $R_h^{(n)}(t)$; blue dashed curve: $R_h^{(n)'}(t)$; black line: the position of finite extreme point of $R_h^{(n)}(t)$.}\label{fig}
\end{figure}

Let us now focus on the feasible set evaluation, which is not empty being
${\by}^{(n)}$ a feasible solution to $\bar{\cal{P}}_p$.  Precisely, omitting the dependence on $d$ and $n$ for notational simplicity and considering the continuous phase case, the feasible set in the transformed domain is given by 
\begin{equation}
\mathcal{F}_\infty=\left(\displaystyle{\cap_{k=1}^K} \mathcal{S}_k\right)\cap\overline{\Psi}_\infty,
\end{equation}
where
\begin{equation}
\mathcal{S}_k = \left\{t: {\bar{a}_{k,h}^{(n)}t^2+\bar{b}_{k,h}^{(n)}t+\bar{c}_{k,h}^{(n)}\le0}\right\},
\end{equation}
and
${\bar\Psi}_{\infty}=[-\tan(\delta/2),\tan(\delta/2)]$.

As  shown in \cite{atsp},	$\mathcal{F}_{\infty}$  can be cast as
\begin{equation}
\mathcal{F}_{\infty}=\displaystyle{\cup_{i=1}^{K_\infty}}\left[l_i,u_i\right],
\end{equation}
with $l_i\le u_i< l_{i+1},i=1,\ldots,K_\infty-1\le K$, whose specific values depend on $\bar{a}_{k,h}^{(n)}$, $\bar{b}_{k,h}^{(n)}$, $\bar{c}_{k,h}^{(n)}$, and $\delta$, {and can be derived relying on  De Morgan law and ``union-find" algorithm.
	As a consequence, the feasible set of Problem ${\cal{P}}_p^{y^{(n)}_h}$ can be expressed as  
	\begin{equation}
	\bar{\mathcal{F}}_{\infty}=\displaystyle{\cup_{i=1}^{K_\infty}}\left[\hat{l}_i^{\infty},\hat{u}_i^{\infty}\right],
	\end{equation}
	where $\hat{l}_i^{\infty}=2\arctan(l_i)$, $\hat{u}_i^{\infty}=2\arctan(u_i)$.
	
	As to the discrete phase case, let us observe that
	\begin{equation}
	\begin{split}
	\bar{\mathcal{F}}_M&=\bar{\mathcal{F}}_{\infty}\cap{\Psi}_M\\
	&=\displaystyle{\bigcup_{i=1}^{K_\infty}}\left\{\left[\hat{l}_i^{\infty},\hat{u}_i^{\infty}\right]\cap{\Psi}_M\right\}.
	\end{split}
	\end{equation}
	Denoting by $\Gamma_i=\left[\hat{l}_i^{\infty},\hat{u}_i^{\infty}\right]\cap{\Psi}_M=\left\{\hat{l}_i,\hat{l}_i+\frac{2\pi}{M},\ldots,\hat{u}_i\right\}$ with
	\begin{eqnarray}
	\hat{l}_i= \left\lceil {\frac{{{\hat{l}_i^{\infty}}M}}{{2\pi }}} \right\rceil \frac{{2\pi }}{M},
	\hat{u}_i=	\left\lfloor {\frac{{{\hat{u}_i^{\infty}}M}}{{2\pi }}} \right\rfloor \frac{{2\pi }}{M},
	\end{eqnarray}
	$\bar{\mathcal{F}}_M$ can be recast as
	\begin{equation}\label{fm}
	\bar{\mathcal{F}}_M={{\mathop  \cup \limits_{\scriptstyle i = 1\atop
				\scriptstyle i: \;\hat{l}_i \le \hat{u}_i }^{K_\infty} \Gamma_i },}
	\end{equation}
	with the condition $\hat{l}_i \le \hat{u}_i$ ensuring that $\Gamma_i$ is not empty. Now actually, denoting by $K_M\le K_\infty$ the number of sets $\Gamma_i$ involved in the union operation of \eqref{fm} (i.e., the number of the actual disjoint closed sets), there exists an increasing mapping $i_t$: $t\in\{1,\ldots,K_M\}\rightarrow i_t\in\{1,\ldots,K_\infty\}$, such that $\bar{\mathcal{F}}_M={{\mathop  \cup \limits_{\scriptstyle t = 1}^{K_M} \Gamma_{i_t} }}$. Hence, denoting by $\Lambda_t=\Gamma_{i_t}=\left\{\hat{l}_t^M,\hat{l}_t^M+\frac{2\pi}{M},\ldots,\hat{u}_t^M\right\}$, $\bar{\mathcal{F}}_M$ can be rewritten as
	\begin{equation}
	\begin{split}
	\bar{\mathcal{F}}_M&={{\mathop  \cup \limits_{\scriptstyle t = 1}^{K_M} \Lambda_t }}.
	\end{split}
	\end{equation}
}	

\begin{lemma}\label{lemma}
	If $\hat{d}_h^{(n)}\neq0$, Eq. \eqref{deri} admits two roots.
\end{lemma}

\begin{proof}
	Let us proceed by contradiction and assume $\Delta_h^{(n)}\le0$. If $\hat{d}_h^{(n)}\neq0$, which entails that $a_{2,h}^{(n)}>0$, two situations may occur:
	\begin{enumerate}
		\item
		$\hat{d}_h^{(n)}>0$, which implies that $R_h^{(n)'}(t)\geq0$ $\forall t$; thus, $R_h^{(n)}(t)$ is a non-decreasing function and in particular,	$\mathop {\lim }\limits_{t \to -\infty } R_h^{(n)}(t)\le R_h^{(n)}(0)\le\mathop {\lim }\limits_{t \to  +\infty } R_h^{(n)}(t)$.
		\item $\hat{d}_h^{(n)}<0$, leading to $R_h^{(n)'}(t)\le0$ $\forall t$; hence,  $R_h^{(n)}(t)$ is  a non increasing function implying that $\mathop {\lim }\limits_{t \to  -\infty } R_h^{(n)}(t)\ge R_h^{(n)}(0)\ge\mathop {\lim }\limits_{t \to +\infty } R_h^{(n)}(t)$.
	\end{enumerate}
	Based on \eqref{rt}, 
	$\mathop {\lim }\limits_{t \to  + \infty } R_h^{(n)}(t) = \mathop {\lim }\limits_{t \to - \infty } R_h^{(n)}(t)=\frac{{a_{1,h}^{(n)}}}{{a_{2,h}^{(n)}}}$, that according to the aforementioned monotonicities leads to $R_h^{(n)}(t)=\frac{{a_{1,h}^{(n)}}}{{a_{2,h}^{(n)}}}$,  $\forall t$.
	Hence,  $R_h^{(n)'}(t)=0$, $\forall t$, namely, $\hat{d}_h^{(n)}=\hat{e}_h^{(n)}=\hat{f}_h^{(n)}=0$,  which contradicts the assumption $\hat{d}_h^{(n)}\neq0$.	
\end{proof}

\subsection{Proof of Proposition \ref{pro}}\label{p1}
\begin{proof}
	To solve Problem ${\cal{P}}_{\infty}^{y^{(n)}_h}$, it is convenient to distinguish among different situations according to the specific instances of $\bar{\mathcal{F}}_{\infty}$ and $\phi_{g,h}^{(n)}$.
	Evidently, if $\phi_{g,h}^{(n)}$ is feasible, i.e., $\phi_{g,h}^{(n)}\in\bar{\mathcal{F}}_{\infty}$, of course $y_h^\star=\phi_{g,h}^{(n)}$.
	
	Let us now suppose that $\phi_{g,h}^{(n)}$ is outside of $[\hat{l}_1^\infty,\hat{u}_{K_\infty}^\infty]$,
	\begin{itemize}
		\item if $-\pi<\phi_{g,h}^{(n)}\le\hat{l}_1^\infty$,  the optimal solution depends on the actual monotonicities of the objective function:
		\begin{itemize}
			
			\item if $\hat{d}_h^{(n)}=0$ and $\hat{e}_h^{(n)}<0$, or $\hat{d}_h^{(n)}<0$,
			$\chi\left(y_h;{\by}_{-h}^{(n)}\right)$ monotonically decreases over $\phi_{g,h}^{(n)}\le y_h<\pi$, which implies $y_h^\star=\hat{l}_1^\infty$;
			\item if  $\hat{d}_h^{(n)}>0$, 		
			$\chi\left(y_h;{\by}_{-h}^{(n)}\right)$ monotonically decreases over $\phi_{g,h}^{(n)}\le y_h\le \phi_{s,h}^{(n)}$, and then strictly increases over $\phi_{s,h}^{(n)}< y_h< \pi$, then $y_h^\star\in\{\hat{l}_1^\infty,\hat{u}_{K_\infty}^\infty\}$;
		\end{itemize} 
		\item if $\phi_{g,h}^{(n)}\ge\hat{u}_{K_\infty}^\infty$, a situation dual to $\phi_{g,h}^{(n)}\le\hat{l}_1^\infty$ occurs:
		\begin{itemize}
			\item if $\hat{d}_h^{(n)}=0$ and $\hat{e}_h^{(n)}<0$, or $\hat{d}_h^{(n)}>0$, 
			$\chi\left(y_h;{\by}_{-h}^{(n)}\right)$ monotonically increases over $-\pi\le y_h\le\phi_{g,h}^{(n)}$, which implies $y_h^\star=\hat{u}_{K_\infty}^\infty$;
			\item if  $\hat{d}_h^{(n)}<0$, 		
			$\chi\left(y_h;{\by}_{-h}^{(n)}\right)$ monotonically decreases over $-\pi\le y_h\le \phi_{s,h}^{(n)}$, and then increases over  $\phi_{s,h}^{(n)}< y_h\le \phi_{g,h}^{(n)}$, thus $y_h^\star\in\{\hat{l}_1^\infty,\hat{u}_{K_\infty}^\infty\}$.
		\end{itemize}
		\item  if $\phi_{g,h}^{(n)}=-\pi$, two situations may occur:
		\begin{itemize}
			\item if $\hat{d}_h^{(n)}=0$ and $\hat{e}_h^{(n)}=0$,  
			$\chi\left(y_h;{\by}_{-h}^{(n)}\right)$ is a constant for all $y_h$, thus $y_h^\star=\hat{l}_1^\infty$;
			\item if $\hat{d}_h^{(n)}=0$ and $\hat{e}_h^{(n)}>0$,		
			$\chi\left(y_h;{\by}_{-h}^{(n)}\right)$ monotonically decreases over $\phi_{g,h}^{(n)}\le y_h\le \phi_{s,h}^{(n)}$, and then increases over $\phi_{s,h}^{(n)}< y_h< \pi$, then $y_h^\star\in\{\hat{l}_1^\infty,\hat{u}_{K_\infty}^\infty\}$;
		\end{itemize}
	\end{itemize}
	Hence in this case, the optimal solution to Problem ${\cal{P}}_{\infty}^{y^{(n)}_h}$ is $y_h^\star=\arg\max\limits_{y_h\in\{\hat{l}_{1}^\infty,\hat{u}_{K_\infty}^\infty\}}\chi\left(y_h;{\by}_{-h}^{(n)}\right)$.
	
	Finally, if $\phi_{g,h}^{(n)}\in[\hat{l}_1^\infty,\hat{u}_{K_\infty}^\infty]$ but $\phi_{g,h}^{(n)}\notin\bar{\mathcal{F}}_{\infty}$, $\phi_{g,h}^{(n)}$
	belongs to $]\hat{u}_{r^\star}^\infty,\hat{l}_{r^\star+1}^\infty[$ with $r^\star$ the highest index such that $\hat{u}_{r^\star}^\infty	< \phi_{g,h}^{(n)}$. It follows that
	\begin{itemize}
		\item if $\hat{d}_h^{(n)}=0$ and $\hat{e}_h^{(n)}<0$, $\chi\left(y_h;{\by}_{-h}^{(n)}\right)$ monotonically increases over $-\pi\le y_h\le \phi_{g,h}^{(n)}$, then decreases over $\phi_{g,h}^{(n)}< y_h< \pi$, thus $y_h^\star\in\left\{\hat{u}_{r^\star}^\infty,\hat{l}_{r^\star+1}^\infty\right\}$;
		\item if $\hat{d}_h^{(n)}>0$, $\chi\left(y_h;{\by}_{-h}^{(n)}\right)$ monotonically increases over $\hat{l}_1^\infty\le y_h\le \phi_{g,h}^{(n)}$, and it is quasi-convex over $[\phi_{s,h}^{(n)},\pi[$, thus $y_h^\star\in\left\{\hat{u}_{r^\star}^\infty,\hat{l}_{r^\star+1}^\infty,\hat{u}_{K_\infty}^\infty\right\}$;
		\item if $\hat{d}_h^{(n)}<0$, $\chi\left(y_h;{\by}_{-h}^{(n)}\right)$ monotonically decreases over $\phi_{g,h}^{(n)}\le y_h\le \hat{u}_{K_\infty}^\infty$, and it is quasi-convex over $[-\pi,\phi_{g,h}^{(n)}]$,  thus $y_h^\star\in\left\{\hat{l}_1^\infty,\hat{u}_{r^\star}^\infty,\hat{l}_{r^\star+1}^\infty\right\}$.
	\end{itemize}
	As a result, in general terms, $y_h^\star\in\left\{\hat{l}_1^\infty,\hat{u}_{r^\star}^\infty,\hat{l}_{r^\star+1}^\infty,\hat{u}_{K_\infty}^\infty\right\}$.
	
\end{proof}
\subsection{Proof of Corollary \ref{fd}}\label{p2}
\begin{proof}
	Two situations may occur: there exists
	an index $q^\star\in\{1,\ldots,K_M\}$ satisfying 
	$\hat{l}_{q^\star}^M\le\phi_{g,h}^{(n)}\le\hat{u}_{q^\star}^M$ or such an index does not exist.
	
	As to the former case,
	$\phi^M_l=\left\lfloor {\frac{{{\phi_{g,h}^{(n)}}M}}{{2\pi }}} \right\rfloor \frac{{2\pi }}{M}$ and $\phi^M_u=\left\lceil {\frac{{{\phi_{g,h}^{(n)}}M}}{{2\pi }}} \right\rceil \frac{{2\pi }}{M}$ represent the  closest feasible points from below and from above, respectively, to  $\phi_{g,h}^{(n)}$. To proceed further, let us consider the optimal solution to the relaxed problem
	
	\begin{equation}\label{rp1}
	\left\{ \begin{array}{ll}
	\displaystyle{\max_{y_h}} & \chi\left(y_h;{\by}_{-h}^{(n)}\right)\\
	\mbox{s.t.}  &   y_h \in [\hat{l}_{1}^M,\phi^M_l]\bigcup[\phi^M_u,\hat{u}_{K_M}^M].
	\end{array} \right.
	\end{equation} 
	It follows that
	\begin{itemize}
		%	\item 
		%	{\color{blue}if $\phi_{g,h}^{(n)}=\phi_{s,d}^{(n)}$ (i.e., $\hat{d}_h^{(n)}=0$ and $\hat{e}_h^{(n)}=0$), $\chi\left(y_h;{\by}_{-h}^{(n)}\right)$ is a constant for all $y_h$, thus let $y_h^\star=\hat{l}_1^\infty$;} 
		\item 
		if $\phi_{g,h}^{(n)}>\phi_{s,d}^{(n)}$ (i.e., either $\hat{d}_h^{(n)}=0$ and $\hat{e}_h^{(n)}<0$, or $\hat{d}_h^{(n)}<0$), $\chi\left(y_h;{\by}_{-h}^{(n)}\right)$ monotonically decreases over $\phi_{g,h}^{(n)}\le y_h< \pi$, thus $\phi^M_u$ is the maximum point over $[\phi^M_u,\pi[$; 
		
		\begin{itemize}
			\item	if $\phi_{s,d}^{(n)}\le\hat{l}_{1}^M$, $\chi\left(y_h;{\by}_{-h}^{(n)}\right)$ monotonically increases over $\hat{l}_{1}^M\le x_d\le\phi^M_l$, thus $\phi^M_l$ is the maximizer over 		
			$[\hat{l}_{1}^M,\phi^M_l]$, implying that the optimal solution to \eqref{rp1} is $y_h^\star=\arg\max\limits_{y_h\in\left\{\phi^M_l,\phi^M_u\right\}}\chi\left(y_h;{\by}_{-h}^{(n)}\right)$;
			\item	instead if $-\pi<\hat{l}_{1}^M<\phi_{s,d}^{(n)}$, $\chi\left(y_h;{\by}_{-h}^{(n)}\right)$ monotonically decreases over $-\pi\le x_d\le\hat{l}_{1}^M$, thus $\chi(\hat{l}_{1}^M;{\by}_{-h}^{(n)})\le\chi(-\pi;{\by}_{-h}^{(n)})=\chi(\pi;{\by}_{-h}^{(n)})\le\chi(\phi^M_u;{\by}_{-h}^{(n)})$, implying again that the optimal solution is $y_h^\star=\arg\max\limits_{y_h\in\left\{\phi^M_l,\phi^M_u\right\}}\chi\left(y_h;{\by}_{-h}^{(n)}\right)$;	
		\end{itemize}
		\item if $\phi_{g,h}^{(n)}<\phi_{s,d}^{(n)}$, following a line of reasoning similar to that for $\phi_{g,h}^{(n)}>\phi_{s,h}^{(n)}$, it can be concluded that $y_h^\star=\arg\max\limits_{y_h\in\left\{\phi^M_l,\phi^M_u\right\}}\chi\left(y_h;{\by}_{-h}^{(n)}\right)$.
	\end{itemize}

	Let us now focus on the latter situation. In this case let us consider the following relaxed problem,  
	\begin{equation}\label{rp}
	\left\{ \begin{array}{ll}
	\displaystyle{\max_{y_h}} & \chi\left(y_h;{\by}_{-h}^{(n)}\right)\\
	\mbox{s.t.}  &   y_h \in \displaystyle{\cup_{i=1}^{K_M}}\left[\hat{l}_i^M,\hat{u}_i^M\right],
	\end{array} \right.
	\end{equation}
	
	Now, according to Proposition III.1, 
	\begin{itemize}
		\item if $\phi_{g,h}^{(n)}\notin[\hat{l}_1^M,\hat{u}_1^M]$, the optimal solution to Problem \eqref{rp} is  $y_h^\star=\max\limits_{y_h\in\{\hat{l}_1^M,\hat{u}_{K_M}^M\}}\chi\left(y_h;{\by}_{-h}^{(n)}\right)$.
		\item otherwise, denoting by $r^\star=\max\limits_{\hat{u}_{r}^M
			< \phi_{g,h}^{(n)}} r$, the optimal solution to \eqref{rp} is  $y_h^\star=\arg\max\limits_{y_h \in\{\hat{l}_1^M,\hat{u}_{r^\star}^M,\hat{l}_{r^\star+1}^M,\hat{u}_{K_M}^M\}}\chi\left(y_h;{\by}_{-h}^{(n)}\right)$.
	\end{itemize}
	Being the optimal solution to the relaxed versions of ${\cal{P}}_M^{y_h^{(n)}}$ in Problems \eqref{rp1} and \eqref{rp} feasible to Problem ${\cal{P}}_M^{y_h^{(n)}}$, the proof is concluded.
\end{proof}

\subsection{Derivation of the Surrogate Function $\tilde{f}_1(\bx_1;\bq^{(l)})$ }\label{33}
Let us observe that the objective function of Problem (30) restricted to $\bq_1={\bx}_1$ is given by
\begin{equation*}\label{ia}
	\hat{f}([\bx_1^T,\bq^{T}_{2}]^T)=\frac{{{{\bx}_1^\dag}{{ \bM}_1^{\left(\bq_2\right)}}{{\bx}_1}}}{{{{{\bx}_1^\dag}}{{\overline \bM}_2^{(\bq_2)}}{{\bx}_1} }} -\bar\beta{\bx}_1^\dag{\bar\bR}{\bx}_1 
\end{equation*}
where ${{\overline \bM}_2^{(\bq_2)}}={\bM}_2^{\left(\bq_2 \right)}+{\vartheta^{\left( \bq_2\right)} }\bI_N/N$, and ${\bar\bR}=\diag\{\bs_0\}^\dag\bR\diag\{\bs_0\}$, which can be expressed as $\hat{f}([\bx_1^T,\bq^{T}_{2}]^T)=g(\bx_1,\bx_1\bx_1^\dag,\bq_{2})-\bar\beta\bx_1^\dag{\bar\bR}\bx_1 $, with $g(\bx_1,\bx_1\bx_1^\dag,\bq_{2})=\frac{{{{{\bx}_1^\dag}}{{ \bM}_1^{\left(\bq_2\right)}}{\bx_1}}}{{\tr\left( {{\overline \bM}_2^{\left(\bq_2\right)}}{\bx_1\bx_1^\dag} \right) }}$. 

Being $g(\bx_1,\bX_1,\bq_2)$ jointly convex with respect to $\bx_1$ and $\bX_1$,  where ${ \bM}_1^{\left(\bq_2\right)},{ \bM}_2^{\left(\bq_2\right)},\bX_1$ are positive definite matrices, a tight expansion is provided by its first-order approximation around any given point $(\hat\bq_1,\hat{\bQ}_1)$, which yields \cite{b2}
\begin{align*}
	\begin{split}
		&u\left(\bx_1,\bX_1;\hat\bq_1,\hat{\bQ}_1,\bq_2\right)\\= &2\Re\left\{\frac{{{{\hat\bq_1}^{\dag}}{{ \bM}_1^{\left(\bq_2 \right)}}{\bx_1}}}{{\tr\left( {{\overline \bM}_2^{(\bq_2)}}{\hat{\bQ}_1} \right) }}\right\}-\frac{{{{\hat\bq_1}^{\dag}}{{ \bM}_1^{\left(\bq_2 \right)}}{\hat\bq_1}}}{{\tr\left( {{\overline \bM}_2^{(\bq_2)}}{\hat{\bQ}_1} \right)^2 }}
		\tr\left({{\overline \bM}_2^{(\bq_2)}} \bX_1 \right).
	\end{split}
\end{align*}
Denoting by $\bar{u}(\bx_1;\bq_1,\bq_2)=u(\bx_1,\bx_1\bx_1^\dag;\bq_1,\bq_1\bq_1^\dag,\bq_{2})$, it follows that
\begin{itemize}
	\item $\hat{f}([\bx_1^T,\hat{\bq}^{T}_{2}]^T)\ge \bar{u}(\bx_1;\hat{\bq}_1,\hat{\bq}_{2})-\bar\beta\bx_1^\dag{\bar\bR}\bx_1,\forall\bx_1,\hat\bq_1,\hat{\bq}_{2};$
	\item $\hat{f}([\bx_1^T,\hat{\bq}^{T}_{2}]^T)= \bar{u}(\bx_1;\bx_1,\hat{\bq}_{2})-\bar\beta{\bx}_1^\dag{\bar\bR}{\bx}_1$.
\end{itemize}

Now, letting  $v(\bx_1,\hat{\bq}_1,\hat{\bq}_{2})=\bar{u}(\bx_1;\hat{\bq}_1,\hat{\bq}_{2})-\bar\beta\bx_1^\dag{\bar\bR}\bx_1$, after some algebraic manipulations it follows that
\begin{equation*}
	\begin{split}
		v(\bx_1;\hat{\bq}_1,\hat{\bq}_{2})
		=-{\bx}_1^\dag\bP^{(\hat{\bq}_1,\hat{\bq}_{2})}{\bx}_1+\Re\left\{\bv^{{(\hat{\bq}_1,\hat{\bq}_{2})}\dag}{\bx}_1\right\},
	\end{split}
\end{equation*}
where 
\begin{equation*}\label{p}
	\begin{split}
		\bP^{(\hat{\bq}_1,\hat{\bq}_{2})}=\frac{{{\hat{\bq}_1^{\dag}}{{ \bM}_1^{\left(\hat{\bq}_{2}\right)}}{\hat{\bq}_1}}}{{\left( {\hat{\bq}_1}^{\dag}{{\overline \bM}_2^{(\hat{\bq}_{2})}}{\hat{\bq}_1} \right)^2 }}{{\overline \bM}_2^{(\hat{\bq}_{2})}}+{\bar\beta}{\bar\bR},
	\end{split}
\end{equation*}
\begin{equation*}
	\bv^{(\hat{\bq}_1,\hat{\bq}_{2})}=2\frac{{{{ \bM}_1^{\left(\hat{\bq}_{2}\right)}}{\hat{\bq}_1}}}{{ \hat{\bq}_1^{\dag}{{\overline \bM}_2^{(\hat{\bq}_{2})}}\hat{\bq}_1 }},
\end{equation*}
Finally, being
\begin{itemize}
	\item ${\bx}_1^\dag\bP^{(\hat{\bq}_1,\hat{\bq}_{2})}{\bx}_1\le\lambda{\bx}_1^\dag{\bx}_1+2\Re\left\{{\bx}_1^\dag\left(\bP^{(\hat{\bq}_1,\hat{\bq}_{2})}-\lambda\bI \right)\hat{\bq}_1 \right\}-
	\hat{\bq}_1 ^{\dag} \left(\bP^{(\hat{\bq}_1,\hat{\bq}_{2})}-\lambda\bI \right)\hat{\bq}_1$, $\forall \bx_1,\hat\bq_1\in{\cal{D}}_p^N$ and $\hat{\bq_2}\in\mathbb{C}^N$, with  equality at ${\bx}_1=\hat{\bq}_1$ and $\lambda=\lambda_{\max}\left( \bP^{(\hat{\bq}_1,\hat{\bq}_{2})}\right)$\cite{b2},
	\item $\left\| {\bx}_1\right\| ^2=N$,
\end{itemize}
it yields
\begin{equation}
\tilde{f}_1({\bx}_1;\hat\bq)=\Re\left\{\bz^{(\hat{\bq}_1,\hat{\bq}_{2})\dag}{\bx}_1\right\}+ r^{(\hat{\bq}_1,\hat{\bq}_{2})}, 
\end{equation}
with $\hat\bq=[\hat{\bq}_1^T,\hat{\bq}_2^T]^T$, and
\begin{eqnarray}\label{zl}
\bz^{(\hat{\bq}_1,\hat{\bq}_{2})}&=&2\left(\lambda\bI-\bP^{(\hat{\bq}_1,\hat{\bq}_{2})} \right)\hat{\bq}_1 +\bv^{(\hat{\bq}_1,\hat{\bq}_{2})},\\
r^{(\hat{\bq}_1,\hat{\bq}_{2})}&=&\hat{\bq}_1 ^{\dag} \left(\bP^{(\hat{\bq}_1,\hat{\bq}_{2})}-\lambda\bI \right)\hat{\bq}_1 -\lambda N,
\end{eqnarray}
which is a valid surrogate function of $\hat{f}([\bq_1^T,\bq^{T}_{2}]^T)$, since it satisfies  the properties (P1)-(P3).

\subsection{The Solution of $q_i^{(l+1)\star}$ in HIVAC}\label{hivac}
As to the continuous phase codes, candidate optimal solutions are the boundary points satisfying the first order optimality condition, i.e., nulling the derivative of the objective function.

To proceed further, let us observe that 
\begin{equation} \label{64}
\left\{ \begin{array}{ll}
\displaystyle{ \max _{x_{i}} } &\tilde f_i\left( x_{i};\bq^{(l)}\right)\\
\mbox{s.t.}  
&x_{i} \in {\cal{D}}_p
\end{array} \right. 
\end{equation}
is equivalent to 
\begin{equation} \label{65}
\left\{ \begin{array}{ll}
\displaystyle{ \max _{\phi_{i}} } &h(\phi_i)\\
\mbox{s.t.}  
&\phi_i \in {\Psi}_p
\end{array} \right. 
\end{equation}
with $h(\phi_i)=\tilde f_i\left( e^{j\phi_i};\bq^{(l)}\right)$. Indeed, denoting by $\phi_i^\star$ the optimal solution to \eqref{65}, $e^{j\phi_i^\star}$ is an optimal solution to \eqref{64}. Let us now observe that 
\begin{equation*}
	\begin{split}
		\frac{dh(\phi_i)}{d\phi_i}&=\frac{d h(2\arctan(\tan(\phi_i/2)))}{d \phi_i}\\
		&=\frac{d \bar{f}(\xi)}{d \xi}\frac{d \xi}{d \phi_i}\\&=\frac{d \bar{f}(\xi)}{d \xi}\frac{1}{2\cos^2(\phi_i/2)},
	\end{split}
\end{equation*} 
with $\bar{f}(\xi)=h(2\arctan(\xi))$ and ${\xi=\tan(\phi_i/2)}$, whose closed form expression is
\begin{equation}\label{ddd1}
\bar f(\xi)=\displaystyle{\frac{u_1^{(l)}\xi^2+v_1^{(l)}\xi+w_1^{(l)}}{u_2^{(l)}\xi^2+v_2^{(l)}\xi+w_2^{(l)}}+\frac{u_3^{(l)}\xi^2+v_3^{(l)}\xi-u_3^{(l)}}{1+\xi^2} } ,
\end{equation}
with\footnote{Note that $\bar{f}(-\infty)=\bar{f}(\infty)=\tilde{f}(e^{-j\pi};\bq^{(l)})=h(-\pi)$.}
$u_1^{(l)}=b_s^{(l)}-\Re\{a_s^{(l)}\}$, $u_2^{(l)}=d_s^{(l)}-\Re\{c_s^{(l)}\}$, $u_3^{(l)}=-\Re\{f_s^{(l)}\}$, $v_1^{(l)}=-2\Im\{a_s^{(l)}\}$, $v_2^{(l)}=-2\Im\{c_s^{(l)}\}$, $v_3^{(l)}=-2\Im\{f_s^{(l)}\}$, $w_1^{(l)}=b_s^{(l)}+\Re\{a_s^{(l)}\}$, $w_2^{(l)}=d_s^{(l)}+\Re\{c_s^{(l)}\}$.

As a result, the stationary points, belonging to $[-\pi,\pi[$, of the objective in \eqref{65} can be computed solving the equation
\begin{equation}\label{fde}
\begin{split}
\frac{d \bar{f}(\xi)}{d \xi}=&\frac{\bar{d}^{(l)}\xi^2+\bar{e}^{(l)}\xi+\bar{r}^{(l)}}{(u_2^{(l)}\xi^2+v_2^{(l)}\xi+w_2^{(l)})^2}-\frac{v_3^{(l)}\xi^2-4u_3^{(l)}\xi-v_3^{(l)}}{(\xi^2+1)^2}\\
=&0.
\end{split}
\end{equation}
where $\bar{d}^{(l)}=u_1^{(l)}v_2^{(l)}-u_2^{(l)}v_1^{(l)}$, $\bar{e}^{(l)}=2(u_1^{(l)}w_2^{(l)}-u_2^{(l)}w_1^{(l)})$, $\bar{r}^{(l)}=v_1^{(l)}w_2^{(l)}-v_2^{(l)}w_1^{(l)}$.

After some algebraic manipulation, it is not difficult to show that \eqref{fde} is tantamount to solving
\begin{equation}\label{6order}
\zeta_6\xi^6+ \zeta_5\xi^5+\zeta_4\xi^4+\zeta_3\xi^3+\zeta_2\xi^2+\zeta_1\xi^1+\zeta_0=0,
\end{equation}
where $\zeta_6=-u_2^{(l)2}v_3^{(l)}+\bar{d}^{(l)}$, $\zeta_5=\bar{e}^{(l)}+4u_2^{(l)2}u_3^{(l)}-2u_2^{(l)}v_2^{(l)}v_3^{(l)}$, $\zeta_4=\bar{r}^{(l)}+2\bar{d}^{(l)}+u_2^{(l)2}v_3^{(l)}+8u_2^{(l)}u_3^{(l)}v_2^{(l)}-v_3^{(l)}(2u_2^{(l)}w_2^{(l)}+v_2^{(l)2})$, $\zeta_3=2\bar{e}^{(l)}+2u_2^{(l)}v_2^{(l)}v_3^{(l)}+4u_3^{(l)}(2u_2^{(l)}w_2^{(l)}+v_2^{(l)2})-2v_2^{(l)}v_3^{(l)}w_2^{(l)}$, $\zeta_2=2\bar{r}^{(l)}+\bar{d}^{(l)}+v_3^{(l)}(2u_2^{(l)}w_2^{(l)}+v_2^{(l)2})+8v_2^{(l)}w_2^{(l)}u_3^{(l)}-v_3^{(l)}w_2^{(l)2}$, $\zeta_1=\bar{e}^{(l)}+2v_2^{(l)}w_2^{(l)}v_3^{(l)}+4u_3^{(l)}w_2^{(l)2}$, $\zeta_0=\bar{r}^{(l)}+w_2^{(l)2}v_3^{(l)}$.
The  real roots of \eqref{6order} are at most six and can be obtained by Matlab ``roots" function. 

Hence, denoting by $ \mathcal{T}_\infty=\{\xi_1,\ldots,\xi_{T_\infty}\},T_\infty\le 6$, the set of real roots of \eqref{6order} belonging to $[-\tan(\delta/2),\tan(\delta/2)]$ and
\begin{equation}\label{ma}
\bar{\mathcal{T}}_{\infty,i}^{(l)}=\{-\delta,\delta,2\arctan(\xi_1),\ldots,2\arctan(\xi_{T_\infty})\},
\end{equation}
the optimal solution to \eqref{65} is given by
\begin{equation}\label{dd1}
\phi_i^\star=\arg\max_{\phi\in\bar{\mathcal{T}}_{\infty,i}^{(l)}}\tilde {f}_i(e^{j\phi};\bq^{(l)}).
\end{equation}

For the finite alphabet case, 
candidate  optimal solutions are the feasible points closest to the stationary points from 
below and from above, respectively.

{Denoting  by $\mathcal{T}_M=\{m_1,\cdots,m_{T_M}\}$, $T_M\le6$, the real roots of \eqref{6order} with $m_i\in[\tan(\frac{\pi\alpha_\epsilon}{M}),\tan(\frac{\pi(\alpha_\epsilon+\omega_\epsilon-1)}{M})],i=1,\ldots,T_M$,  the stationary points of $h(\phi_i)$ are $\{\iota_1,\ldots,\iota_{T_M}\}$ with $\iota_i=2\arctan(m_i),i=1,\ldots,T_M$.
	Then,  the sets of feasible points in \eqref{65} closest
	to $\iota_i,i=1,\ldots,T_M$ from 
	below and from above are $\mathcal{T}_1=\{\theta_i,i=1,\cdots,T_M|\theta_i=\left\lfloor {\frac{{\iota_i}M}{{2\pi }}} \right\rfloor \frac{{2\pi }}{M}\}$, and $\mathcal{T}_2=\{\bar{\theta}_i,i=1,\ldots,T_M|\bar{\theta}_i=\left\lceil {\frac{{\iota_i}M}{{2\pi }}} \right\rceil \frac{{2\pi }}{M}\}$, respectively.
	
}
Hence, letting
\begin{equation}\label{ma2}
\bar{\mathcal{T}}_{M,i}^{(l)}=\{\frac{2\pi\alpha_\epsilon}{M},\frac{2\pi(\alpha_\epsilon+\omega_\epsilon-1)}{M},\mathcal{T}_1\bigcup\mathcal{T}_2\},
\end{equation} 
the optimal solution to \eqref{65} for the finite alphabet case is given by\footnote{  Note that if ${\omega_{\epsilon}}\le 14$,   the direct search may require a lower computational complexity.}
\begin{equation}\label{dd2}
\phi_i^\star=\arg\max_{\phi\in\bar{\mathcal{T}}_{M,i}^{(l)}}\tilde {f}_i(e^{j\phi};\bq^{(l)}).
\end{equation}

%%%%%%%%%%%%%%%%%%%%%%%%%%%%%%%%%%%%%%%%%%%%%%%%%%%%%%%%%%%%%%%%%%%%%%%%%%%%%%%%%%%%%%%%%%%%%%%%%%%%
%%%%%%%%%%%%%%%%%%%%%%%  REFERENCE SECTION (Bibliography) %%%%%%%%%%%%%%%%%%%%%%%%%%%%%%%%%%%%%%%%%%

%
\bibliographystyle{IEEEbib}
\bibliography{IEEEabrv,biblio_paper2}
\end{document}